\documentclass[11pt]{article}
\usepackage{hyperref}
\usepackage{tikz} 

\usepackage{kpfonts}

\usepackage{tikz-cd}

\usepackage{amsmath,amsthm,amsfonts,amssymb,amsopn,amscd,bm} 
\usepackage{color}

\DeclareMathOperator{\ran}{ran}

\def\tilde{\widetilde}

\def\a{\alpha}
\def\b{\beta}

\def\k{\kappa}

\def\th{\theta}

\def\A{{\cal A}}
\def\B{{\cal B}}

\def\C{{\cal C}}
\def\D{{\cal D}}

\def\R{{\cal R}}

\def\H{{\cal H}}
\def\K{{\cal K}}
\def\S{{\cal S}}

\def\f{{\varphi}}
\def\s{{\sigma}}

\def\x{{{h}}}

\def\PSL{{{\rm PSL}(2,\mathbb R)}}

\def\S2{S^{1(2)}}

\def\RR{\mathbb R}

\newtheorem{theorem}{Theorem}[section]
\newtheorem{lemma}[theorem]{Lemma}

\newtheorem{corollary}[theorem]{Corollary}

\newtheorem{proposition}[theorem]{Proposition}

\theoremstyle{remark}

\newcommand{\ben}{\begin{equation}}
\newcommand{\een}{\end{equation}}

\newcommand{\bthm}{\begin{theorem}}
\newcommand{\ethm}{\end{theorem}}

\newcommand{\bprop}{\begin{proposition}}
\newcommand{\eprop}{\end{proposition}}
\newcommand{\bcor}{\begin{corollary}}
\newcommand{\ecor}{\end{corollary}}
\newcommand{\blem}{\begin{lemma}}
\newcommand{\elem}{\end{lemma}}

\def\PSL{PSU(1,1)}

\def\CC{{\mathbb C}}

\def\SL2{{{\rm SL}(2,\R)}}

\def\PSL2{{{\rm PSL}(2,\Reali)}}

\def\U1{{{\rm V}(1)}}
\def\SU2{{{\rm SV}(2)}}

\def\SU{{{\rm SU}}}

\def\A{{\mathcal A}}
\def\B{{\mathcal B}}
\def\C{{\mathcal C}}
\def\D{{\mathcal D}}

\def\H{{\mathcal H}}

\def\K{{\mathcal K}}
\def\k{K}

\def\S{{\mathcal S}}

\def\bp{{\bf p}}

\def\cH{{\cal H}}

\def\cL{{\cal L}}

\def\bC{{\mathbb C}}

\def\bR{{\mathbb R}}
\def\RR{{\mathbb R}}

\def\a{\alpha}
\def\b{\beta}

\def\th{\vartheta}    
\def\k{\kappa}

\def\x{\xi}

\def\s{\sigma}

\def\f{\varphi}

\title{\Huge{Modular structure of the Weyl algebra}}

\author{{\sc Roberto Longo\thanks{Supported by the ERC Advanced Grant 669240 QUEST ``Quantum Algebraic Structures and Models'', MIUR FARE R16X5RB55W  QUEST-NET and GNAMPA-INdAM. \eject
E-mail: {longo@mat.uniroma2.it}
}
}\\
Dipartimento di Matematica,
Universit\`a di Roma Tor Vergata,\\
Via della Ricerca Scientifica, 1, I-00133 Roma, Italy
}

\date{}
\begin{document}

\maketitle

\begin{abstract}
We study the modular Hamiltonian associated with a Gaussian state on the Weyl algebra. We obtain necessary/sufficient  criteria for the  local equivalence of Gaussian states, independently of the classical results by Araki and Yamagami, Van Daele, Holevo. 
We then present a criterion for a Bogoliubov automorphism to be weakly inner in the GNS representation. 
We also describe the vacuum modular Hamiltonian associated with a time-zero interval in the scalar, massless, free QFT in two spacetime dimensions, thus complementing the recent results in higher space dimensions \cite{LM}.  
In particular, we have the formula for the local entropy of a one-dimensional massless wave packet and Araki's vacuum relative entropy of a coherent state on a double cone von Neumann algebra.  
\end{abstract}

\newpage
\section{Introduction}
The Heisenberg commutation relations  are at the core of Quantum Mechanics. From the mathematical viewpoint, they have a more transparent formulation in Weyl's exponential form. If $H$ is a real linear space equipped with a non-degenerate symplectic form $\b$, one considers the free ${}^*$-algebra $A(H)$ linearly generated by the (unitaries) $V(h)$, $h\in H$, that satisfy the commutation relations (CCR)
\ben\label{CCR}
V(h+k) = e^{i\b (h,k)}V(h)V(k)\, , \quad h,k\in H\, ,
\een
$V(h)^* = V(-h)$. The Weyl algebra $A(H)$ admits a unique $C^*$ norm, so its $C^*$ completion is a simple $C^*$-algebra, the Weyl $C^*$-algebra $C^*(H)$. The representations, and the states, of $A(H)$ and of $C^*(H)$ are so in one-to-one correspondence. We refer to \cite{BR,Pe,DG} for the basic theory.

For a finite-dimensional $H$, von Neumann's famous uniqueness theorem shows that all representations of $C^*(H)$, with $V(\cdot)$ weakly continuous, are quasi-equivalent. 
As is well known, in Quantum Field Theory (QFT) one deals with infinitely many degrees of freedom and many inequivalent representations arise, see \cite{Haag}.

Due to the relations \eqref{CCR}, a state on $C^*(H)$ is determined by its value on the Weyl unitaries; a natural class of states is given by the ones with Gaussian kernel. A state $\f_\a$ is called Gaussian, or quasi-free, if
\[
\f_\a\big(V(h)\big)  = e^{-\frac12 \a(h,h)}\, ,
\]
with $\a$ a real bilinear form $\a$ on $H$, that has to be compatible with $\b$. 

Assuming now that $H$ is separating with respect to $\a$, as is the case of a local subspace in QFT, the GNS vector associated with $\f_\a$ is cyclic and separating for the von Neumann algebra $\A(H)$ generated by $C^*(H)$ in the representation. 
So there is an associated Tomita-Takesaki modular structure, see \cite{Tak}, that we are going to exploit in this paper. 

The modular theory is a deep, fundamental operator algebraic structure that is widely known and we refrain from explaining it here, giving for granted the reader to be at least partly familiar with that. We however point out two relevant aspects for our work. The first one is motivational and concerns the growing interest on the modular Hamiltonian in nowadays physical literature, especially in connection with entropy aspects (see e.g. refs in \cite{L18}). The other aspect concerns the crucial role taken by the modular theory of standard subspaces, see \cite{LN}; this general framework, where Operator Algebras are not immediately visible, reveals a surprisingly rich structure and is suitable for applications of various kind. Most of our paper will deal with standard subspaces. 

Our original motivation for this paper was the description of the local modular Hamiltonian associated with the free scalar QFT in $1+1$ spacetime dimension, in order to complement the higher dimensional results, that were obtained in \cite{LM}. 
In the end, our massless formula in Section \ref{MHd1}, is obtained directly. 
 As a consequence, we compute the local entropy of a low dimensional wave packet. This gives also Araki's vacuum relative entropy  of a coherent state on a local von Neumann algebra the free, massless, scalar QFT, now also in the $1+1$ dimension case. We refer to \cite{L18,L19,CLR19,LM} for background results and explanation of the context. 

We now briefly describe part of the background of our work. The Canonical Commutation Relations \eqref{CCR} and  Anti-Commutation Relations
are ubiquitous and intrinsic in Quantum Physics. The study of the corresponding linear symmetries (symplectic transformations, CCR case) is a natural problem; the automorphisms of the associated operator algebras are called Bogoliubov automorphisms, see \cite{EK,DG}.  The classical result of Shale \cite{Sh} characterises the Bogoliubov automorphisms that are unitarily implementable on the Fock representation. Criteria of unitary implementability in a quasi-free representation were given by Araki and Yamagami \cite{Ar82}, van Daele \cite{vD} and Holevo \cite{Ho}; these works are independent of the modular theory, although the last two rely on the purification construction, that originated in the classical paper by Powers and St\o rmer in the CAR case \cite{PS}. Woronowicz partly related the purification map to the modular theory and reconsidered the CAR case \cite{Wo}.  However, the modular structure of the Weyl algebra has not been fully exploited so far, although the CCR case is natural to be studied from this point of view.  

We work in the context of the standard form of a von Neumann algebra studied by Araki, Connes and Haagerup \cite{Ar74,Con74, Ha75}. If an automorphism of a von Neumann algebra in standard form is unitarily implementable, then it is canonically implementable; so we know where to look for possible implementation. This will provide us with a criterion for local normality that is independent of the mentioned previous criteria, we however make use of Shale's criterion. We shall give necessary/sufficient criteria for the quasi-equivalence of Gaussian states in terms of the modular data. 

A key point in our analysis concerns the cutting projection on a standard subspace studied in \cite{CLR19}. This projection is expressed in terms of the modular data; on the other hand it has a geometric description in the QFT framework. The cutting projection is thus a link between geometry and modular theory, so it provides a powerful tool. 

Among our results, we have indeed necessary/sufficient criteria for the quasi-equi\-va\-lence of two Gaussian states $\f_{\a_1}$, $\f_{\a_2}$ on $C^*(H)$, in terms of the difference of certain functions of the modular Hamiltonians, that are related to the cutting projections. 

The following diagram illustrates the interplay among the three equivalent structures associated with standard subspaces and the geometric way out to QFT:
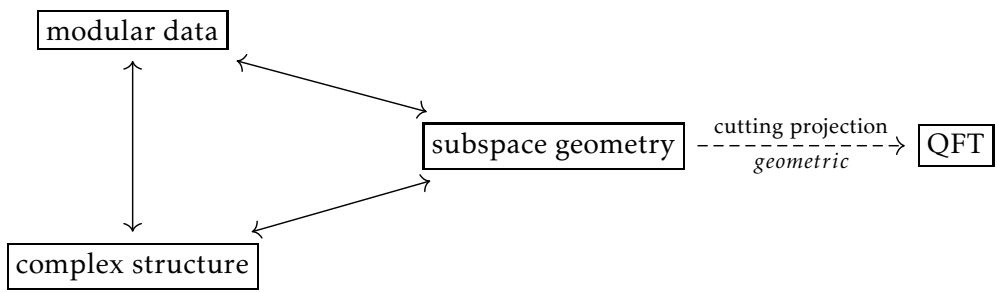
\[
\begin{tikzcd}
\boxed{\text{modular data}}\arrow[drr, leftrightarrow] \arrow[dd, leftrightarrow]  &  & \\  
 & & \boxed{\text{subspace geometry}} \arrow[rrr, dashed, "\text{cutting projection}", "geometric" '] & & & \boxed{\text{QFT}} \\
\boxed{\text{complex structure}}\arrow[urr, leftrightarrow]  & & 
\end{tikzcd} 
\]
Our paper is organised as follows. We first study the modular structure of standard subspaces, especially in relation to polarisers and cutting projections. We then study the local normality/weak innerness of Bogoliubov transformations, and the quasi-equi\-valence of Gaussian states in terms of modular Hamiltonians and other modular data. 

In the last part, we give the formula for the low dimensional modular Hamiltonian in the free massless Quantum Field Theory. We also include appendices of independent interest. 

This manuscript is a shortened version of the published article \cite{LW}, and takes into account the correction pointed out in the erratum. 

\section{Basic structure}
This section contains the analysis of some general, structural aspects related to closed, real linear subspaces of a complex Hilbert space, from the point of view of the modular theory. 

\subsection{One-particle structure}\label{OPS}

Let $H$ be a real vector space. A {\it symplectic form} $\b$ on $H$ is a real, bilinear, anti-symmetric form on $H$. We shall say that $\b$ is {\it non-degenerate} on $H$ if 
\[ 
\ker\b \equiv \{h\in H: \b(h,k) = 0 \, ,\ \forall k\in H\} = \{0\}\, .
\]
We shall say that $\b$ is {\it totally degenerate} if $\ker\b = H$, namely $\b=0$. 
A {\it symplectic space} is a real linear space $H$ equipped with a symplectic form $\b$. 

Given a symplectic space $(H,\b)$, a real scalar product $\a$ on $H$ is \emph{compatible} with $\b$ (or $\b$ is compatible with $\a$) if the inequality 
\ben\label{ba}
 \b(h,k)^2\leq   \a(h,h) \a(k,k) \ , \quad h,k\in H\, ,
\een
holds. Given a compatible $\a$, note that $\ker\b$ is closed (w.r.t. $\a$), $\b= 0$ on $\ker\b$  and $\b$ is non-degenerate on $(\ker\b)^\perp$. 
Clearly, $\b$ extends to a symplectic form on the completion $\bar H$ of $H$ w.r.t. $\a$, compatible with the extension of $\a$. (However $\b$ may be degenerate on $\bar H$ even if $\b$ is non-degenerate on $H$.)

A \emph{one-particle structure} on 
$H$ associated with the compatible scalar product $\a$ (see \cite{Kay}) is  a pair $(\H,\kappa)$, where $\cH$  is a complex Hilbert space and $\kappa: H\to\cH$ is a real linear map satisfying
\begin{itemize}
\item[$a)$] $\Re(\kappa(h_1),\kappa(h_2))= \a(h_1,h_2)$ and $\Im(\kappa(h_1),\kappa(h_2))= \b(h_1,h_2)$, \ $h_1,h_2\in H$,
\item[$b)$] $\kappa(H)+i\kappa(H)$ is dense in $\H$. 
\end{itemize} 
Note that $\k$ is injective because
\ben\label{inj}
h\in H,\ \k(h) = 0 \Rightarrow \Re(\kappa(h),\kappa(h)) = 0 \Rightarrow \a(h,h) = 0 \Rightarrow h = 0\, .
\een
With $\bar H$ the completion of $\bar H$ w.r.t. $\a$, $\b$ extends to a compatible symplectic form on $\bar H$. 
Then $\k$ extends to a real linear map $\bar\k:\bar H\to \H$ with $(\H,\bar\k)$ a one-particle structure for $\bar H$.  

In the following proposition, we shall anticipate a couple of facts explained in later sections. 
The uniqueness can be found in \cite{Kay}; the existence is inspired by \cite{Pe}. 
\bprop\label{onep}
Let $H$ be a symplectic space with a compatible scalar product $\a$.    
There exists a one-particle structure $(\H, \kappa)$ on $H$ associated with $\a$. It is unique, modulo unitary equivalence; namely, if  $(\cH',\k')$ is another one-particle structure on $H$, there exists a unitary $U: \cH\to \cH'$ such that  the following diagram commutes:
\[
 \begin{tikzcd}[row sep=tiny]
& \H \arrow[dd, "U"] \\
H \arrow[ur, "\kappa"]  \arrow[dr, "\kappa'" '] & \\
& \H'
\end{tikzcd}
\]
\eprop
\proof
{\it Uniqueness.} The linear map $U: \k(h) \mapsto \k'(h)$ is well defined on $\k(H)$ by \eqref{inj}. Moreover, it extends to a complex linear map $\kappa(H)+i\kappa(H) \to \k'(H)+i\k'(H)$ 
and is isometric because
\begin{multline*}\label{inj2}
||\k(h) + i \k(k)||^2 =  
||\k(h)||^2 + ||\k(k)||^2 + 2\Re(\k(h), i \k(k))  \\
= ||\k(h)||^2 + ||\k(k)||^2 - 2\Im(\k(h),  \k(k))  = \a(h,h) + \a(k,k) -2\b(h,k)  = ||\k'(h) + i \k'(k)||^2\, ,
\end{multline*}
so $U$ extends to a unitary operator with the desired property. 

{\it Existence.} By replacing $H$ with its completion w.r.t. $\a$, we may assume that $H$ is complete. Suppose first that $\b$ is totally degenerate, i.e. $\b =0$, and let $H_\CC$ the usual complexification of $H$, namely $H_\CC = H \oplus H$ as real Hilbert space with complex structure given by the matrix 
$i  = 
\begin{bmatrix} 0  & -1 \\
1  & 0
\end{bmatrix}$. Then $\k: h\in H\mapsto h\oplus 0\in H_\CC$ is a one-particle structure on $H$ associated with $\a$. 

Suppose now that $\b$ is non-degenerate and consider the polariser $D_H$ (Sect. \ref{polariser}). 
If $\ker(D_H^2 + 1) =\{0\}$, i.e. $H$ is separating (see Lemma \ref{nondl}), the orthogonal dilation provides a one-particle structure on $H$ associated with $\a$ (Sect. \ref{Opur}). If $D^2_H = -1$, then $D_H$ is a complex structure on $H$, so the identity map is a one-particle structure. Taking the direct sum, we see that a 
one-particle structure exists if $\b$ is non-degenerate. 

The existence of a one-particle structure then follows in general because $H = H_a\oplus H_f$, where the restriction of $\b$ to $H_a$ is totally degenerate and to $H_f$ is non-degenerate. 
\endproof

\subsection{Polariser}\label{polariser}
Let $H\subset \cH$ be a closed, real linear subspace of the complex Hilbert space $\H$. By the Riesz lemma, there exists a unique bounded, real  linear operator $D_H$ on $H$ such that
\ben\label{polar}
\b(h,k) = \a(h, D_H k) \, ,\quad h,k\in H\, ,
\een
with $\a(\cdot,\cdot) = \Re(\cdot,\cdot)$, $\b(\cdot,\cdot) = \Im(\cdot,\cdot)$

We have 
\[
||D_H|| \leq 1\, , \quad  D_H^* = - D_H\, .
\]
The operator $D_H$ is called the {\it polariser} of $H$. 
As\[
\Im (h, k) = -\Re (h, ik) = - \Re(h, E_H ik) \, ,\quad h,k\in H\, ,
\]
we have one of our basic relations
\ben\label{P1}
D_H = - E_H i |_H\, ,
\een
where $E_H$ is the orthogonal projection onto $H$. 

Let $H' = (iH)^{\perp_\RR}$ be the {\it symplectic complement} of $H$. We shall say that $H$ is {\it factorial} if $H\cap H'= \{0\}$. 
\blem\label{nondl}
We have
\ben\label{sep}
{\rm ker}(D_H^2 + 1) = H\cap  i H\, ,
\een
thus $H$ is separating iff $\ker(D_H^2 + 1) = \{0\}$. 
Furthermore,
\ben\label{nond}
 {\rm  ker}(D_H) =  \ker\b = H\cap H' \,  . 
\een
thus $H$ is factorial iff {\rm ker}$(D_H) = \{0\}$. 
\elem
\proof
As $D_H = - E_H i |_H$, with $E_H$ the orthogonal projection of $\H$ onto $H$ \eqref{P1}, we have 
\ben\label{D2}
D_H^2 = E_H i E_H i |_H = - E_H E_{iH}|_H
\een
so, if $h\in H$,
\[
(D_H^2 + 1)h= 0  \Leftrightarrow   E_H E_{iH} h =   h \Leftrightarrow h \in H\cap iH\, ,
\]
showing the first part of the lemma. 

Last assertion follows as
\[
\ker\b  = \ran(D_H)^\perp = \ker(D_H^*) = \ker(D_H)
\]
and clearly $\ker \b = H\cap H'$. 
\endproof
\bprop
$h \in \ker(D_H^2 + 1) \Leftrightarrow ||D_H h|| = ||h||  \Leftrightarrow D_H h = -ih$. 
\eprop
\proof
Let $h\in \ker(D_H^2 + 1)$, thus $D_H^2 h = - h$, so $||D^2_H h|| = ||h||$ and this implies $||D_H h|| = ||h||$ because
$||D_H|| \leq 1$. 
Thus $||E_H i h|| = ||h|| = ||ih||$, so $h\in iH$; hence $h\in H\cap  i H$.  So $D_H h = -E_H i h = -ih$. 

Conversely, assume that $D_H h  = -ih$; then $ih\in H$, so $||D_H h|| = ||E_H i h||= ||h||$. 
Finally, assume the equality 
$||D_H h|| = ||h||$ to hold. Then $||E_H i h|| = ||ih||$, so $E_H i h = ih$, hence
$D_H h = -E_Hi h = -ih$, so $D_H^2 = -h$, namely $h \in \ker(D_H^2 + 1)$. 
\endproof

\subsection{Standard subspaces}
Let $\cH$ be a complex Hilbert space and $H$ a closed, real linear subspace. We say that $H$ is \emph{cyclic} if $H + i H$ is dense in $\cH$, \emph{separating} if $H \cap i H = \{0\}$, \emph{standard} if it is both cyclic and separating.  

Let $H\subset \cH$ be a closed, real linear subspace of $\H$ and $\b = \Im(\cdot,\cdot)$ on $H$, where $(\cdot,\cdot)$ is the complex scalar product on $\H$; then $\b$ is a symplectic form on $H$ that makes it a symplectic space. Moreover, 
$\a=\Re(\cdot,\cdot)$ is a compatible real scalar product on $H$. 

An {\it abstract standard subspace} is a triple $(H,\a,\b)$, where $H$ is a real Hilbert space, $\a$ is the real scalar product on $H$ and $\b$ is a symplectic form on $H$ compatible with $\a$,  so that $H$ separating, that is $\ker(D^2_H + 1) = \{0\}$, with $D_H$ the polariser of $H$, see Lemma \ref{nondl}. 

By Proposition \ref{onep}, an abstract standard subspace can be uniquely identified, up to unitary equivalence, with a standard subspace of a complex Hilbert space as above. 

We shall say that the abstract standard subspace $(H,\a,\b)$ is {\it factorial} if $\ker(D_H) = \{0\}$, namely $\b$ is non-degenerate. 

In view of the above explanations, we shall often  directly deal with standard subspaces of a complex Hilbert space $\H$. 

Given a standard subspace $H$ of $\H$,  
we shall denote by $J_H$ and $\Delta_H$ the {\it modular conjugation} and the {\it modular operator} of $H$; they are defined by the polar decomposition $S_H = J_H \Delta_H^{1/2}$ of the closed, densely defined, anti-linear involution on $\H$
\[
S_H : h + ik \mapsto h - ik\, , \quad h,k\in H\, .
\]
$\Delta_H$ is a non-singular, positive selfadjoint operator, $J_H$ is an anti-unitary involution and we have 
\ben\label{JDJ}
J_H\Delta_H J_H = \Delta_H^{-1}\, .
\een 
The fundamental relations are
\[
\Delta_H^{is} H = H\, , \quad J_H H = H'\, , \quad s\in\RR\, ,
\]
see \cite{RV, LRT, LN}. 
We denote by
\[
L_H = \log\Delta_H
\]
the {\it modular Hamiltonian} of $H$. We often simplify the notation setting $ L = L_H$ and similarly for other operators. 

Assume now $H$ to be standard and factorial.
Let  $E_H$ be the real {\it orthogonal projection} from $\H$ onto $H$ as above and $P_H$ the {\it cutting projection}
\ben\label{cutting}
P _H: h + h' \mapsto h\, , \quad h\in H,\, h'\in H'\, .
\een
$P_H:D(P_H)\subset\H\to\H$ is a closed, densely defined, real linear operator with domain $D(P_H) = H + H'$. 

Recall two formulas respectively in \cite{FG} and in \cite{CLR19}:
\begin{align}
&E_H = (1 + \Delta_H)^{-1} + J_H\Delta_H^{1/2} (1 + \Delta_H)^{-1} \label{fE}\, ,\\
&P_H = (1 - \Delta_H)^{-1} + J_H\Delta_H^{1/2} (1 - \Delta_H)^{-1} \label{fP} \, ;
\end{align}
more precisely, $P_H$ is the closure of the right hand side of  \eqref{fP}.

These formulas can be written as
\begin{align}
&E_H = (1 + S_H )(1 + \Delta_H)^{-1} \, ,\label{fE1} \\
&P_H = (1 + S_H ) (1 - \Delta_H)^{-1}  \, , \label{fP1}
\end{align}
so give
\ben\label{PE}
P_H = E_H (1 + \Delta_H) (1 - \Delta_H)^{-1} = - E_H \coth(L_H/2)\, .
\een
In the following, if $T:D(T)\subset \H\to\H$ is a real linear operator, $T|_H$ is the restriction of $T$ to $D(T|_H)\equiv D(T)\cap H$, that we may consider also as operator $H\to H$ if $\ran(T|_H)\subset H$, as it will be clear from the context. 
\bprop\label{pol}
Let $H\subset \H$ be a factorial standard subspace. The polariser $D_H$ of $H$ and its inverse $D_H^{-1}$ are given by
\begin{gather}\label{DH1}
D_H = - E_H i|_H = i( \Delta_H -1)(\Delta_H +1)^{-1}|_H= i\tanh(L_H/2)|_H\, ,
\\
D_H^{-1} = P_H i|_H = - i( \Delta_H +1)(\Delta_H -1)^{-1}|_H = - i\coth(L_H/2)|_H \, .\label{DH2}
\end{gather}
As a consequence, $P_H i|_H$ is a skew-selfadjoint real linear operator on $H$. 
\eprop
\proof
As $J_H\Delta_H J_H = \Delta_H^{-1}$, eq. \eqref{fE} gives
\[
E_H  = (1 + \Delta_H)^{-1} +\Delta_H(1 + \Delta_H)^{-1}  J \Delta_H^{1/2}\, ,
\]
therefore
\begin{multline}
E_H ih  =\Big( (1 + \Delta_H)^{-1} +\Delta_H(1 + \Delta_H)^{-1}  S_H \Big) ih = (1 + \Delta_H)^{-1}ih - \Delta_H(1 + \Delta_H)^{-1} ih\\
= (1 - \Delta_H)(1 + \Delta_H)^{-1}ih \, ,
\end{multline}
$h\in H$, 
thus
\ben\label{E_H i}
E_H i |_{H} = (1 - \Delta_H)(1 + \Delta_H)^{-1}i|_H\, .
\een
As $D_H = - E_H i |_H$ \eqref{P1},  eq. \eqref{DH1} is proved.  

Concerning formula \eqref{DH2}, since $H$ is left invariant by $(\Delta_H + 1)(\Delta_H -1)^{-1}i$, 
from \eqref{PE} we get 
\[
P_H i|_H = - E_H  \coth(L_H/2)i|_H = - i \coth(L_H/2)|_H
= - i(\Delta_H + 1)(\Delta_H -1)^{-1}|_H \, .
\]
So $P_H i|_H$ is skew-selfadjoint because $H$ is globally $\Delta_H^{is}$-invariant, 
$s\in \RR$ \cite[Prop. 2.2]{LM}. 
\endproof
\bcor\label{D1}
We have
\ben\label{D12}
\sqrt{1 + D_H^{2}} = 2(\Delta_H^{1/2} + \Delta_H^{-1/2})^{-1}|_H  =  \frac1{\cosh (L_H/2)}\Big|_H\, .
\een
\ben\label{D11}
D_H^{-1}\sqrt{1 + D_H^{2}} = - 2i(\Delta_H^{1/2} - \Delta_H^{-1/2})^{-1}|_H  = - i \frac1{\sinh (L_H/2)}\Big|_H\, ;
\een
\ecor
\proof
By Prop. \ref{pol} $D_H = i\tanh(L_H/2)|_H$, thus 
\ben\label{1D2}
D_H^2 = -\tanh^2(L_H/2)|_H\, ,
\een
 so $D_H^2$ is a bounded selfadjoint operator on $H$ (as real linear operator). Therefore
\ben\label{1D3}
1 + D_H^2 = \big(1 -\tanh^2(L_H/2)|_H\big)\big|_H  =  \frac1{\cosh^2(L_H/2)}\Big|_H \, ,
\een
thus \eqref{D12} holds. 

By Prop. \ref{pol} we then have
\[
D_H^{-1}\sqrt{1 + D_H^{2}}  
= - i \frac{\coth(L_H/2) }{\cosh (L_H/2)}\Big|_H  
= - i \frac1{\sinh (L_H/2)}\Big|_H\, .
\]
\endproof
The following corollary follows at once from \cite{LX21}. The type of a subspace refers to the second quantisation von Neumann algebra. 
\bcor\label{corEED}
We have
\ben\label{EDD}
E_H E_{H'}|_H = 1 + D_H^{2} \, .
\een
Therefore, $H$ is a type $I$ subspace iff $1 + D_H^{2}$ is a trace class operator. 
\ecor
\proof
By \cite[Lemma 2.4]{LX21}, we have $E_H E_{H'}|_H = 4\Delta_H(1 + \Delta_H)^{-2}|_H$;
by \eqref{1D3}, we have 
\[
4\Delta_H(1 + \Delta_H)^{-2}|_H = \frac1{\cosh^2(L_H/2)}\Big|_H =
 1 + D_H^{2}\, .
\] 
The corollary thus follows by \cite[Cor. 2.6]{LX21}. 
\endproof
By \eqref{EDD} and \eqref{D2}, we have the nice identity
\ben
E_H E_{H'}|_H + E_H E_{iH}|_H = 1 \, .
\een
Let $(H,\a_k, \b)$ be abstract standard subspaces, $k=1,2$, and suppose that $\a_1$ is equivalent to $\a_2$, thus there exists a bounded, positive linear map $T: H\to H$ with bounded inverse such that $\a_2(h,k) = \a_1(h,Tk)$. Then
\[
\a_1(h, D_1 k) = \b(h,k) = \a_2(h, D_2 k)  = \a_1(h, T D_2 k)\, ,
\]
thus $D_1 = T D_2$. 

\subsection{Orthogonal dilation}\label{Opur}
Let $H$ be a real Hilbert space, with real scalar product $\a$, and consider the doubling
\[
\tilde H = H \oplus H
\]
(direct sum of real Hilbert spaces).
We consider a symplectic form $\b$ on $H$, that we assume to be non-degenerate and compatible with $\a$.  Let $D$ be the polariser of $\b$ on $H$ given by \eqref{polar}. So $\ker(D) = \{0\}$. We also assume that $\ker(1 + D^2) =\{0\}$, namely $(H,\a, \b)$ is a factorial abstract subspace \eqref{sep}. 
Set
\ben\label{tD}
\iota   
= \begin{bmatrix} D  & V\sqrt{1 + D^{2}}\, \\
V\sqrt{1 + D^{2}}  & - D\end{bmatrix}
\, ,
\een
with $V$ the phase of $D$ in the polar decomposition, $D = V|D|$; note that $V$ commutes with $D$, because $D$ is skew-selfadjoint, and $V^2 = -1$ (see \cite{Pe,BCD}). 
Then $\iota$ is a unitary on $\tilde H$ and $\iota^2 = -1$, namely $\iota$ is a complex structure on $\tilde H$. 

Let $\cH$ be the complex Hilbert space given by $\tilde H$ and $\iota$. The scalar product of $\H$ is given by
\[
(h_1 \oplus h_2, k_1 \oplus k_2)  
= \tilde\a(h_1 \oplus h_2, k_1 \oplus k_2) + i \tilde\b(h_1 \oplus h_2, k_1 \oplus k_2)
\]
with $\tilde\a \equiv \a\oplus\a$ and $\tilde\b(h_1 \oplus h_2, k_1 \oplus k_2) =
\tilde\a(h_1 \oplus h_2, \iota (k_1 \oplus k_2))$. 

The embedding $\k : H \to \cH$
\[
\k : h \mapsto \kappa(h) \equiv h\oplus 0 
\]
satisfies the condition $b)$ in Sect. \ref{OPS}, that is $\tilde\a(\kappa(h),\kappa(k)) = \a(h,k)$ and
\[
\tilde\b(\kappa(h),\kappa(k)) = \tilde\a(h \oplus 0, \iota (k \oplus 0)) =
\tilde\a(h \oplus 0, Dk \oplus V\sqrt{1 + D^{2}}\, k)) = \a(h, D k) = \b(h,k)\, ,
\]
$h,k\in H$. 
\blem
$\k(H)$ cyclic and separating in $\tilde H$,
so $\k$ is a one-particle structure for $H$ with respect to $\a$ and $\k(H)$ is a factorial subspace. 
\elem
\proof
$\k(H)$ cyclic means that the linear span of $H\oplus 0$ and $\{\iota (h\oplus 0): h\in H\}$ is dense in $\cH$. 
As
\[
\iota (h\oplus 0) = Dh\oplus - V\sqrt{1 + D^{2}}\, h\, ,
\]
$\k(H)$ is cyclic iff $\ran(V\sqrt{1 + D^{2}})$ is dense, thus iff ${\rm ker}(1 + D^2) = \{0\}$. 
The proof is then complete by Lemma \ref{nondl}. 
\endproof
By the above discussion $H\subset \H$  is a factorial standard subspace. 
We call $H\subset\H$ the {\it orthogonal dilation} of $(H,\b)$ with respect to $\a$. 
\subsection{Symplectic dilation}
Let $(H,\a, \b)$ be an abstract factorial standard subspace. Consider the doubled symplectic space
$(H\oplus H ,\hat\b)$,
where $\hat\b = \b\oplus-\b$. 

With $D$ the polariser of $\a$, let $H_0 = \ran(D)$ and set 
\ben\label{tA}
\iota  = 
\begin{bmatrix} D^{-1}  & D^{-1}\sqrt{1 + D^2}\, \\
- D^{-1}\sqrt{1 + D^{2}}  & - D^{-1}\end{bmatrix} \, ,
\een
where the matrix entries are defined as real linear operators $(H,\a) \to (H,\a)$ with domain $H_0$. Then
\[
\iota^2 = -1
\]
on $H_0 \oplus H_0$. 
A direct calculation shows that
\ben\label{ub}
\hat\b(\iota \xi,\iota\eta) \equiv \hat\b( \xi,  \eta)\, ,\quad \xi,\eta \in H_0\oplus H_0\, ;
\een
setting 
\ben\label{ub1}
\hat\a(\xi,\eta) \equiv \hat\b(\xi, \iota \eta)\, ,\quad \xi,\eta \in H_0\oplus H_0\, ,
\een
we have a real scalar product $\hat\a$ on $H_0\oplus H_0$ which is compatible with $\hat \b$. Let $\hat\H$ be the completion of $H_0\oplus H_0$ with respect to $\hat\a$; then $\hat\H$ is a real Hilbert space with scalar product still denoted by $\hat \a$. 

By \eqref{ub}, \eqref{ub1}, $\iota$ preserves $\hat \a$, so the closure of $\iota$ is a complex structure on $H$, and $\iota$ is the polariser of $\hat \a$ w.r.t. $\hat\b$. 
Then $\hat\b$ extends to a symplectic form on $\H$ compatible with $\hat\a$. 
So $\hat\H$ is indeed a complex Hilbert space and $H\subset\hat\H$ is a real linear subspace, where $H$ is identified with $H\oplus 0$. 

We call $H\subset\hat\H$ the {\it symplectic dilation} of $(H,\b)$ with respect to $\a$. 
\bprop
$H$ is a factorial standard subspace of the symplectic dilation $\hat\H$. Therefore the symplectic and the orthogonal dilations are unitarily equivalent. 
\eprop
\proof
$H$ is complete, thus closed in $\hat\H$. Since the polariser of $H$ in $\hat\H$ is equal to $D$, the proposition follows by Lemma \ref{nondl}. 
\endproof

\section{Bogoliubov automorphisms}
\label{BW}
In this section we study symplectic maps that promote to unitarily implementable automorphisms on the Fock space. 

Given a symplectic space $(H, \b)$,  we consider the {\it Weyl algebra} $A(H)$ associated with $H$, namely the free 
${}^*$-algebra complex linearly generated by the Weyl unitaries $V(h)$, $h\in H$, that satisfy the commutation relations
\[
V(h+k) = e^{i\b (h,k)}V(h)V(k)\, , \quad V(h)^* = V(-h)\, , \ \quad h,k\in H\, .
\]
The $C^*$ envelop of $A(H)$ is the {\it Weyl $C^*$-algebra}. 
If $\b$ non-degenerate, there exists a unique $C^*$ norm on $A(H)$ and $C^*(H)$ is a simple $C^*$-algebra. 

Let $\H$ be a complex Hilbert space and $e^\H$ be the Bosonic Fock Hilbert space over $\H$. Then we have the {\it Fock representation} of $C^*(\H_\RR)$ on $e^\H$, where $\H_\RR$ is $\H$ as a real linear space, equipped with the symplectic form $\b \equiv \Im(\cdot,\cdot)$. In the Fock representation, the Weyl unitaries are determined by their action on the vacuum vector $e^0$
\ben\label{Weyl}
V(h)e^{0} = e^{-\frac12 (h,h)}e^{h}\, ,\ h\in \H\, ,
\een
where $e^h$ is the coherent vector associated with $h$.  So the {\it Fock vacuum state} $\f = (e^0,\cdot\, e^0)$ of $C^*(\H_\RR)$ is given by
\ben\label{Fock}
\f\big(V(h)\big)= e^{-\frac12 ||h||^2}\, ,\quad h\in \H\, .
\een
With $H$ any real linear subspace of $\H$, the Fock representation determines a representation of $C^*(H)$ on $e^\H$, which is cyclic on $e^\H$ iff $H$ is a cyclic subspace of $\H$.  We denote by $\A(H)$ the von Neumann algebra on $e^\H$ generated by the image of $C^*(H)$ in this representation. We refer to \cite{BR,MV,L,LN} for details. 

\subsection{Global automorphisms}\label{GA}
Let $\H$ be a complex Hilbert space and $e^\H$ the Fock space as above. 
A {\it symplectic map} $T: D(T)\subset \H\to \H$ is a real linear map with $D(T)$ and ran$(T)$ dense, that preserves the imaginary part of the scalar product, thus $\Im (T\xi, T\eta) = \Im(\xi,\eta)$, $\xi,\eta\in D(T)$.  

Let $T: D(T)\subset \H\to \H$ be a symplectic map. Then 
\[
\Re( iT\xi, T\eta) = \Re(i\xi, \eta)\, , \quad \xi,\eta\in D(T)\, ,
\]
thus $iT\xi\in D(T^*)$ and $T^*iT\xi = i\xi$ for all $\xi\in D(T)$, namely
\ben\label{T*}
T^*iT = i|_{D(T)} \, ,
\een
therefore ker$(T) = \{0\}$, $T$ is closable because $T^*$ is densely defined, and $T^{-1} = -iT^*i|_{\ran(T)}$, so $T^*|_{i\ran(T)}$ is a symplectic map too. It also follows that
\ben\label{Tb}
T \ \text{bounded} 
\Longleftrightarrow T^* \ \text{bounded} \Longleftrightarrow T^{-1} \ \text{bounded} \, . 
\een
We then have the associated {\it Bogoliubov homomorphism} $\th_T$ of the Weyl algebra $A\big(D(T)\big)$ 
onto $A\big(\!\ran(T)\big)$:
\[
\th_T : V(\xi) \mapsto V(T\xi)\, , \ \xi \in D(T)\, .
\]
Let $T:\H\to\H$ be a bounded, everywhere defined symplectic map; the criterion of Shale \cite{Sh} gives a necessary and sufficient condition in order that $\th_T$ be unitary implementable on $e^\H$, under the assumption that $T$ has a bounded inverse:
\ben\label{Shale}
\th_T\ \text{unitary implementable}\ \Longleftrightarrow\ T^*T - 1 \in \cL^2(\H)\ \Longleftrightarrow\ [T,i] \in \cL^2(\H)\, ,
\een
where $[T,i] = Ti -iT = Ti(1-T^*T)$ is the commutator and $\cL^2(\H)$ are the real linear, Hilbert-Schmidt operator on $\H$. 

Due to the equivalence \eqref{Tb}, the assumption $T^{-1}$ bounded in \eqref{Shale} can be dropped (as we assume that ran$(T)$ is dense). 

We shall deal with symplectic maps that, a priori, are not everywhere defined. However the following holds. 
\begin{lemma}\label{ul}
Let $T: D(T)\subset \H\to\H$ be a symplectic map. Then $\th_T$ is unitarily implementable 
iff $\th_{\overline T}$ is unitarily implementable, where $\overline T$ is the closure of $T$.  In this case, $T$ is bounded. 
\end{lemma}
\begin{proof}
 First we show that, if $\th_T$ is implemented by a unitary $U$ on $e^\H$, then $T$ is bounded. Indeed, if $\xi_n\in D(T)$ is a sequence of vectors with $\xi_n\to 0$, then $V(\xi_n) \to 1$ strongly, thus $V(T\xi_n) =UV(\xi_n)U^* \to 1$, so
\[
\f\big((V(T\xi_n)\big) = e^{-\frac12 ||T \xi_n||^2} \to 1\, ,
\] 
with $\f$ the Fock vacuum state, therefore $||T \xi_n||\to 0$ and $T$ is bounded.  

If $\th_{\overline T}$ is implemented, then $\th_{T}$ is obviously implementable by the same unitary. 
Conversely, assume that $\th_T$ is implementable by a unitary $U$ on $\H$. So $T$ is bounded. Hence $\overline T$ is a bounded, everywhere defined symplectic map. Let $\xi\in \H$ and choose a sequence of elements $\xi_n\in D(T)$ such that $\xi_n \to \xi$. Then
\[
\th_{\overline T}\big(V(\xi)\big) = V({\overline T}\xi) = \lim_n V(T \xi_n) = \lim_n UV(\xi_n)U^* = UV(\xi)U^*\, ,
\]
so
$\th_{\overline T}$ is implemented by $U$.
\end{proof}

\subsection{Hilbert-Schmidt perturbations}
Motivated by Shale's criterion, we study here Hilbert-Schmidt conditions related to the symplectic dilation of a symplectic map. 

We use the following 
{\it notations:} If $\H$ is a complex Hilbert space, $\cL^p(\H)$ denotes the space of real linear, densely defined operators $T$ on $\H$ that are bounded and the closure $\bar T$ belongs to the Schatten $p$-ideal with respect to the real part of the scalar product, $1\leq p<\infty$. 
If $\H_1,\H_2$ are complex Hilbert spaces, $T\in\cL^p(\H_1,\H_2)$ means $T^* T\in\cL^{\frac{p}2}(\H_1)$. 
If $H\subset \H$ is a standard subspace, $T\in \cL^p(H)$ means that $T$ is a real linear, everywhere defined operator on $H$ in the Schatten $p$-ideal with respect to the real part of the scalar product. Similarly, $T\in \cL^p(H_1,H_2)$ means $T\in \cL^{\frac{p}2}(H)$. 

Let now $H\subset\H$ be a factorial standard subspace of the Hilbert space $\H$ and $C: H + H' \to H + H'$ a real linear operator. As $H + H'$ is the linear direct sum of $H$ and $H'$, we may write $C$ as a matrix of operators
\ben\label{matrixsympl}
C  = 
\begin{bmatrix} C_{11}  & C_{12}\, \\
C_{21}  &  C_{22} \end{bmatrix}
\een
(the {\it symplectic matrix decomposition}). 
Thus
\[
C_{11} = P_H C |_H, \ \ C_{12} = P_H C |_{H'}, \dots 
\]
and $C_{11}$ is an operator $H\to H$, $C_{12}$ is an operator $H'\to H$, etc. 

We want to study the Hilbert-Schmidt condition for $C$. Note that
\[
C\in\cL^2(\H) \Longleftrightarrow E_H C E_H \in\cL^2(\H),\ \ E_H  C E_{H^\perp} \in\cL^2(\H)\, \dots 
\]
With $D= D_H$ the polariser and $J = J_H$ the modular conjugation,
the symplectic matrix decomposition of the complex structure is
\ben\label{ipur}
i  = 
\begin{bmatrix} D^{-1} & D^{-1}\sqrt{1 + D^{2}}J
\, \\
- J D^{-1}\sqrt{1 + D^{2}}  &  - J D^{-1} J \end{bmatrix}\, ,
\een
as follows from \eqref{tA} and the uniqueness of the dilation. 
Note, in particular, the identity
\ben\label{P'i}
P_{H'}i|_H    =  - J D^{-1}\sqrt{1 + D^{2}} \, .
\een
\blem
The following symplectic matrix representations hold:
\[
E_H =
\begin{bmatrix} 1  & \sqrt{1 + D^{2}}\,J \\
0  &  0 \end{bmatrix}\, ,\quad 
E_{H^\perp} = 
\begin{bmatrix} 0  & -\sqrt{1 + D^{2}}J  \\
0  &  1 \end{bmatrix}\, ,\quad
E_{H'} =
\begin{bmatrix} 0 & 0 \\
J \sqrt{1 + D^{2}}  & 1 \end{bmatrix}\, .
\]
\elem
\proof
We have
\ben\label{EHi}
E_H i 
= \begin{bmatrix} - D  & 0\, \\
0  &  0 \end{bmatrix}
\een
because $E_H i$ is equal to $-D$ on $H$ and zero on $H' = iH^\perp$. As
$E_H = -(E_H i)i $, the first equality in the lemma follows by matrix multiplication with \eqref{ipur}. 
The second equality is then simply obtained as
\[
E_{H^\perp} = 1 - E_H =
\begin{bmatrix} 0  & -\sqrt{1 + D^{2}}J  \\
0  &  1 \end{bmatrix} \, .
\]
Last equality follows as
\[
E_{H'} = J E_H J
\]
and the symplectic matrix decomposition of $J$ is $\begin{bmatrix} 0  & J  \\
J  &  0 \end{bmatrix}$. 
\endproof
\blem\label{En0}
Let $C: H + H' \to H+ H'$ be a real linear map such that $i Ci =  C$, with symplectic matrix decomposition \eqref{matrixsympl}. We have
\begin{gather}\label{E10}
E_H C |_H  =  C_{11} + \sqrt{1 + D^{2}}\, J C_{21}\, ,
\\
\label{E30}
E_{H} C i|_{H'} =D C_{12}\, ,
\\
\label{E20}
 E_{H'}i C |_H =   J D J  C_{21}\, ,
 \\
 \label{E40}
E_{H'} C |_{H'}  =  J \sqrt{1 + D^{2}}\,  C_{12} +  C_{22}\, .
\end{gather}
\elem
\proof
We have
\ben\label{EHC0}
E_H C  = \begin{bmatrix}  C_{11} + \sqrt{1 + D^{2}}\,J C_{21}
 & C_{12} + \sqrt{1 + D^{2}}\, J C_{22}
\, \\
0 &   0 \end{bmatrix} \, ,
\een
thus
\[
E_H C |_H  =  C_{11} + \sqrt{1 + D^{2}}\, J C_{21}\ ,
\]
namely, \eqref{E10} holds. 

Since $Ci = -i C$, we have
\[
E_{H} C i = - E_{H} i C 
=  \begin{bmatrix} D & 0 \\
0 & 0 \end{bmatrix}C \, ,
\]
so
\[
E_{H} C i = \begin{bmatrix} D C_{11}  & D C_{12} \\
0 & 0 \end{bmatrix} \, ,
\]
thus
\[
E_{H} C i|_{H'} =D C_{12}
\]
and \eqref{E30} holds. 

With $C^j = J C J$, we then get
\begin{multline*}
E_{H'}i C |_H 
= J E_H J iC  |_H 
= - J E_H J C  i |_H 
= - J E_H  C^j J i|_H 
=  J E_H  C^j i J  |_H \\
= J (E_H C^j i ) |_{H'} J 
= J D C^j_{12} J 
= J D J J C^j_{12} J 
= J D J  C_{21} \, , 
\end{multline*}
so \eqref{E20} holds. 

Similarly, from \eqref{E10} we get \eqref{E40}. 
\endproof
With $H$ a standard subspace, {\it a symplectic map of the standard subspace $H$} is a real linear map $T: H\to H$ such that
\[
\Im(Th,Tk) = \Im(h,k)\, ,\quad h,k\in H\, ,
\]
equivalently
\[
\Re(Th,D Tk) = \Re(h, Dk)\, ,\quad h,k\in H\, , 
\]
so 
\[
T\ \text{symplectic} \Leftrightarrow  T^* D T = D\, ;
\]
if $T$ is invertible, we shall say that $T$ is a {\it symplectic bijection of $H$}. 

Now, let $H$ be a factorial standard subspace and $T: H\to H$ a symplectic bijection.  
Denote by $\tilde  T$ the symplectic map $T \oplus J T J: H + H'\to H + H'$, namely
$\tilde T = TP_H + JTJ P_{H'}$, i.e. 
\[
\tilde T =
\begin{bmatrix} T &   0\\
0&   J T J \end{bmatrix}
\]
in the symplectic matrix description. We have
\[
\tilde T i = \begin{bmatrix} T D^{-1}&   T D^{-1}\sqrt{1 + D^{2}}\, J\\
- J T  D^{-1}\sqrt{1 + D^{2}}&   - J T D^{-1}J \end{bmatrix}\, ,
\]
\[
i\tilde T = \begin{bmatrix}  D^{-1}T&    D^{-1}\sqrt{1 + D^{2}}\, T J\\
- J D^{-1}\sqrt{1 + D^{2}}\,  T&   - J D^{-1} T J\end{bmatrix}\, ,
\]
\[
[\tilde T,i]  = \begin{bmatrix} [T,D^{-1}]&  \big[T, D^{-1}\sqrt{1 + D^{2}}\big] J\\
- J \big[T, D^{-1}\sqrt{1 + D^{2}}\big]&  - J [T,D^{-1}] J \end{bmatrix}\, .
\]
Note that 
\[
i[\tilde T,i]i  =  i(\tilde  T i - i \tilde T)i = - i \tilde T + \tilde T i  = [\tilde T,i]\, .
\]
\bcor\label{En}
We have
\begin{gather}\label{E1}
E_H [ \tilde T,i ] |_H = [T,D^{-1}] - \sqrt{1 + D^{2}}\,\big[T, D^{-1}\sqrt{1 + D^{2}}\big]\, ,
\\
\label{E3}
E_{H}  \big[ \tilde T,i \big] i|_{H'} = D\big[T, D^{-1}\sqrt{1 + D^{2}}\big] J \, ,
\\
\label{E2}
 E_{H'}i \big[ \tilde T,i \big] |_H 
=
-J D\big[T, D^{-1}\sqrt{1 + D^{2}}\big]\, ,
\\
\label{E4}
E_{H'} \big[ \tilde T,i \big] |_{H'}  = 
J \big(\sqrt{1 + D^{2}}\,\big[T, D^{-1}\sqrt{1 + D^{2}}\big] -  [T,D^{-1}]\big) J\, .
\end{gather}
\ecor
\proof We apply Lemma \ref{En0} with $C = \big[ \tilde T,i\big]$. 
By \eqref{E10}, we get  \eqref{E1}. 
By \eqref{E30}, we get \eqref{E3}. 
By \eqref{E20}, we get \eqref{E2}. 
By \eqref{E40}, we get \eqref{E4}. 
\endproof
\bprop\label{norm}
$[\tilde T, i]\in \cL^2(\H)$ iff both the following conditions hold:
\begin{gather*}
a)\ [T,D^{-1}] - \sqrt{1 + D^{2}}\,\big[T, D^{-1}\sqrt{1 + D^{2}}\big]\in \cL^2(H)\, ,\\ 
b)\ D\big[T, D^{-1}\sqrt{1 + D^{2}}\big]\in \cL^2(H)\, .
\end{gather*}
\eprop
\proof
Assume $[\tilde T, i]\in \cL^2(\H)$. Then the operators \eqref{E1}, \eqref{E3} are Hilbert-Schmidt, 
and this implies that the operators in the statement are in $\cL^2(H)$.  

Conversely, assume that the operators in the statement are in $\cL^2(H)$. Then the operators in Lemma \ref{En} are in $\cL^2(H)$. 

Now,
\[
E_{H^\perp} C E_{H^\perp} = i E_{H'}i C i E_{H'}i = - i E_{H'} C E_{H'}i \, ,
\]
thus
\[
E_{H^\perp} C |_{H^\perp} \in \cL^2(H^\perp) \Longleftrightarrow E_{H'} C |_{H'} \in \cL^2(H')\, ;
\]
moreover,
\ben\label{ET3bis}
E_{H} C |_{H^\perp}   \in \cL^2(H^\perp,H) \Longleftrightarrow E_{H} C i|_{H'}   \in \cL^2(H', H)\, .
\een
We conclude that all the four matrix elements in the orthogonal decomposition of $\big[\tilde T, i\big]$ are in $\cL^2(\H)$, thus
$[\tilde T, i]\in \cL^2(\H)$. 
\endproof
\bcor
Assume $[T, D^{-1}]\in \cL^2(H)$ and $\big[T, D^{-1}\sqrt{1 + D^{2}}\big] \in \cL^2(H)$. Then
$[\tilde T, i]\in \cL^2(\H)$. 
\ecor
\proof
If the assumptions are satisfied, then $a)$ and $b)$  of Prop. \ref{norm} clearly hold because $D$ and $\sqrt{1 + D^{2}}$ are  bounded.  
\endproof

\subsection{Local automorphisms}
Let now $H_k$ be standard factorial subspaces of the Hilbert spaces $\H_k$, $k=1,2$ and $T: H_1 \to H_2$ a symplectic bijection, namely $T$ is real linear, invertible and $\b_2(Th, Tk) = \b_1(h,k)$, $h,k\in H_1$, with $\b_k$ the symplectic form on $H_k$ (the restriction of $\Im(\cdot, \cdot)_k$ to $H_k$, with $(\cdot, \cdot)_k$ the scalar product on $\H_k$). Then $T$ promotes to a $^*$-isomorphism $\th_T$ between the Weyl $C^*$-algebras $C^*(H_1)$ and $C^*(H_2)$
\[
\th_T\big(V_1(h)\big) = V_2(Th)\, .
\]
With $\A_k(H_k)$ the von Neumann algebra associated with $H_k$ on the Bose Fock space $e^{\H_k}$, we want to study when $\th_T$ extends to a normal isomorphism between $\A_1(H_1)$ and $\A_2(H_2)$. 

Let $\tilde T: \H_1\to \H_2$ be the real linear operator, with domain $D(\tilde T) = H _1+ H'_1$ and range $\ran(\tilde T) = H_2 + H'_2$,
\[
\tilde T: h + J_1 k \mapsto Th + J_2 Tk\, ,\quad h,k\in H_1\, ,
\]
where $H'_k$ is the symplectic complement of $H_k$ in $\H_k$ and $J_k = J_{H_k}$. 
Then $\tilde T$ is a densely defined, real linear, symplectic map with dense range from $\H_1$ to $\H_2$.  
\blem\label{bdd}
If $\tilde Ti_1 - i_2 \tilde T$ is bounded and densely defined, then $\tilde T$ is bounded. 
\elem
\proof
$\tilde T$ is closable by Lemma \ref{ul} so $\tilde T i_1$ and  $i_2 \tilde T$ are closable too.
By assumptions, there is a bounded, everywhere defined operator $C:\H_1\to\H_2$ such that $\tilde T i_1 = i_2 \tilde T + C$ 
on $\D \equiv D(\tilde T i_1 - i_2 \tilde T)$,
so the closures of $\tilde T i_1|_\D$ and  $i_2 \tilde T|_\D$ have the same domain.   
Now
\[
\D = D(\tilde T) \cap i_1 D(\tilde T) = D(P_{H_1}) \cap i_1 D(P_{H_1})
\]
is a core for $P_{H_1}$, as follows by eq. \eqref{fP}. Indeed, $\Delta_{i_1 H_1} = \Delta_{H_1}$ and $J_{i_1 H_1} = - J_{H_1}$, so the spectral subspaces of $\Delta_{H_1}$ relative to finite closed intervals $[a,b]\subset (0,1)\cup (1,\infty)$ are in the domain of $D(P_{H_1})\cap D(P_{i_1{H_1}})$ (see \cite{CLR19}). 

Now,
\[
\tilde T = T P_{H_1} + J_2TJ_1(1 - P_{H_1}) 
\]
and one easily checks that $\D$ is a core for $\tilde T$, similarly as above. It follows that $\bar{\tilde T}i_1 = i_2\bar{\tilde T} + C$, with $\bar{\tilde T}$ the closure of $\tilde T$. 
Therefore,
$D(\bar{\tilde T}i_1)= D(i_2\bar{\tilde T})$, so $i_1 D(\bar{\tilde T}) = D(\bar{\tilde T})$. We conclude that 
\[
D(\bar{\tilde T})\supset (H_1 + H_1')+ i_1(H_1 + H_1')\supset H_1 + i_1H_1' = H_1 + H_1^\perp = \H_1\, ,
\]
so $\tilde T$ is bounded by the closed graph theorem. 
\endproof
\bprop\label{HS}
The following are equivalent:
\begin{itemize}
\item[$(i)$] There exists a unitary $U: e^{\H_1}\to e^{\H_2}$ such that $UV_1(h) U^* = V_2(Th)$, $h \in H_1$;

\item[$(ii)$] $\th_T$ extends to a normal isomorphism $\A_1(H_1) \to \A_2(H_2)$;

\item[$(iii)$]  $\tilde T^* \tilde T - 1 \in \cL^2(\H_1)$;

\item[$(iv)$] $\tilde Ti_1 - i_2  \tilde T\in \cL^2(\H_1,\H_2)$. 
\end{itemize}
\eprop
\proof
$(i) \Leftrightarrow (ii)$: Clearly $(ii)$ follows from $(i)$; we show that 
$(ii) \Rightarrow(i)$.  Let $V_k(\cdot)$ be the Weyl unitary on $e^{\H_k}$. By assumptions, the linear extension of the map $V_1(h)\mapsto V_2(Th)$, $h\in H_1$, extends to a normal isomorphism $\bar\th_T:\A_1(H_1)\to \A_2(H_2)$. Since the vacuum vector is cyclic and separating for $\A_k(H_k)$, we have the associated unitary standard implementation $U_T: e^{\H_1}\to e^{\H_2}$ of $\bar\th_T$ w.r.t. the vacuum vectors \cite{Ar74,Con74,Ha75}. 

$(i) \Leftrightarrow(iii)$: Assume $(i)$ and let $U_T$ be the vacuum unitary standard implementation  $\bar\th_T$
as above. $e^{J_k}$, the second quantisation of the modular conjugation $J_k$ of $H_k$, is the modular conjugation of the von Neumann algebra $\A_k(H)$ w.r.t. the vacuum vector $e^0$, so
we have
\[
U_T V_1(h) U_T^* = V_2(Th)\, ,\quad U_T e^{J_1} = e^{J_2} U_T \, , ,\quad h\in H_1\, ,
\]
therefore
\[
U_T V_1(h) V_1(J_1 k) U_T^* = V_2(h) V_2(J_2 k)\, ,\quad h,k\in H_1\, ,
\]
namely
\[
U_T V_1(h + J_1 k) U_T^* = V_2(Th + J_2 Tk)\, ,
\]
that is 
\ben\label{imp}
\ U_T V_1(\eta) U_T^* = V_2(\tilde T\eta)\, ,
\een
for all $\eta$ in the domain of $\tilde T$. Then $(iii)$ holds by Lemma \ref{ul} and Shale's criterion \cite{Sh}. Conversely, assuming $(iii)$, by Lemma \ref{bdd} and again by Lemma \ref{ul} and Shale's criterion, we can find a unitary $U$ such that \eqref{imp} holds. 

$(iii)$ and $(iv)$ are equivalent, by using Lemma \ref{ul} and Lemma \ref{bdd}, see e.g. \cite{LM}. 
\endproof
\bcor\label{norm2}
Let $T: H_1\to H_2$ be a symplectic bijection. Then the Bogoliubov isomorphism $\th_T: A(H_1) \to A(H_2)$ is implemented by a unitary $U: e^{\H_1} \to e^{\H_2}$ iff the following conditions hold:
\begin{gather*}
a)\ \Big(TD_1^{-1} - D_2^{-1}T\Big) 
- \sqrt{1 + D_2^{2}}\,\Big(TD_1^{-1}\sqrt{1 + D_1^{2}} - D_2^{-1}\sqrt{1 + D_2^{2}}\, T\Big)
 \in \cL^2(H_1, H_2)\\ 
b)\ D_2\Big(TD_1^{-1}\sqrt{1 + D_1^{2}}- D_2^{-1}\sqrt{1 + D_2^{2}}\, T\Big) \in \cL^2(H_1, H_2)\, .
\end{gather*}
\ecor
\proof
The above conditions are the straightforward generalisations of the conditions $a)$ and $b)$ in Proposition \ref{norm}, so the corollary follows by Proposition \ref{HS}. 
\endproof
Recall that a real linear map $T: H_1 \to H_2$ is symplectic iff $T^*D_2  = D_1 T^{-1}$,
so the conditions in the above corollary take a different form by inserting this relation. 

\section{Gaussian states, modular Hamiltonian, quasi-equivalence}\label{QF}
Let $(H,\b)$ be a symplectic space. With $\a$ a real scalar product on $H$ compatible with $\b$, 
let $\kappa_\a: H\to \H_\a$ be  the one-particle structure associated with $\a$ (Prop. \ref{onep}).

Let $e^{\H_\a}$ be the Bose Fock Hilbert space over $\H_\a$ and denote by $V_\a(\cdot)$ the Weyl unitaries acting on 
$e^{\H_\a}$ and by $e^0$ the vacuum vector of $e^{\H_\a}$, thus $V(h)\mapsto V_\a(h)$ gives a representation of $C^*(H)$ on $e^{\H_\a}$ (see for example \cite{L}). 
By \eqref{Fock}, we have
\ben\label{qfs0}
(e^0, V_\a(\kappa_\a(h))e^0) =  e^{-\frac12 ||\kappa_\a(h)||^2}
= e^{-\frac12 \a(h,h)} \, ,\quad h\in H\, .
\een
\bprop\label{qfs}
There exists a unique  state $\f_\a$ on $C^*(H)$ such that
\ben\label{qfdef}
\f_\a\big(V(h)\big) = e^{-\frac12 \a(h,h)}\, .
\een
With $\{\H_{\f_\a}, \pi_{\f_\a}, \xi_{\f_\a}\}$ the GNS triple associated with $\f_\a$, the vector $\xi_{\f_\a}$ is separating for the von Neumann algebra $\A(H)=\pi_{\f_\a}\big(C^*(H)\big)''$ iff the completion $\bar H$ of $H$ is a separating subspace, namely $\ker(D_{\bar H}^2 + 1) = \{0\}$. 
\eprop
\proof
Eq. \eqref{qfs0} shows that there exists a state $\f_a$ such that \eqref{qfdef} holds. Moreover \eqref{qfdef} determines $\f_\a$ because the linear span of the Weyl unitaries is a dense subalgebra of $C^*(H)$. 

As $\kappa_\a(H)$ is cyclic in $\H_\a$, $\overline{\kappa_\a(H)}$ is a standard subspace of $\H_\a$ iff $\overline{\kappa_\a(H)}$ is separating. On the other hand, $e^0$ is cyclic and separating 
for the von Neumann algebra generated by the $V_\a(h)$'s, $h\in H$, iff  $\overline{\kappa_\a(H)}$ is a standard subspace of $\H$, see \cite{L}. The proposition then follows by the uniqueness of the GNS representation. 
\endproof
The state $\f_\a$ determined by \eqref{qfdef} is well known and is called the {\it Gaussian}, or {\it quasi-free, state} associated with $\a$, see \cite{Pe, DG}. It is usually defined by showing directly, by positivity, that the Gaussian kernel \eqref{qfdef} defines a state. 

We summarise in the following diagram the two above considered, unitarily equivalent constructions with the GNS representation of a Gaussian state:
\[
\begin{tikzcd}
(H,\a,\b) \arrow[r, "\!\text{Weyl}", "\b" '] 
&C^*(H)\arrow[r,"\!\text{GNS}", "\f_\a" '] 
&  \H_{\f_\a}, \xi_{\f_\a} 
\\  
h\in H\arrow[urr, "\!\!\!\!\!\pi_{\f_\a}(V(h))" ' near end, bend right=7,mapsto]  
\arrow[u, hook] \arrow[d, hook]\arrow[drr, "V_\a(h)" near end, bend left=7, mapsto]  &  {}  &{}  
\\
(H,\a,\b)  \arrow[r, "\text{1-p. str.}", "\kappa_\a" '] & \H_\a \arrow[r, "\!\!\text{Fock}"] & e^{\H_\a}, e^0 \arrow[uu, leftrightarrow, dashed ]
\end{tikzcd} 
\]
As a consequence, if $H$ is a standard subspace, the modular group $\s^{\f_\a}$ of $\f_\a$ on $C^*(H)$ is given by
\[
\s_s^{\f_\a}\big(V(h)\big) = V\big(\Delta_H^{is}h\big)\, ,\quad h\in H\, , \ s\in\RR\, ,
\]
therefore the study of the modular structure of $\A(H)$ can be reduced to the study of the modular structure of $H$. 

The following quasi-equivalence criterion is related to the analysis in \cite{Ar82,vD,Ho}, although we do not rely on their work. 

In the following, we shall always deal with {\it factorial standard subspaces}. 
\bthm\label{basic}
Let $(H, \a_k, \b)$ be factorial, abstract standard subspaces, $k=1,2$.  
The Gaussian states $\f_{\a_1}$ and $\f_{\a_2}$ are quasi-equivalent iff both
\ben\label{t1}
(D^{-1}_1 -  D^{-1}_2)  - \sqrt{1 + D_2^{2}}\,\Big(D^{-1}_1\sqrt{1 + D_1^{2}} - D^{-1}_2\sqrt{1 + D_2^{2}}\Big)
\in \cL^2(H)
\een
and
\ben\label{t2}
D_2 \Big(D^{-1}_1\sqrt{1 + D_1^{2}} - D^{-1}_2\sqrt{1 + D_2^{2}}\Big) \in \cL^2(H)\, ,
\een
hold, where $D_k$ is the polariser of $(H, \a_k, \b)$. 
\ethm
\proof
Let $\H_k$ be the symplectic dilation of $(H, \b_k)$ with respect to $\a_k$; so $H\subset \H_k$ is a factorial standard subspace. We have spelled out the conditions for the symplectic map $ I: \hat H \to \hat H$ to promote a unitary between the Fock spaces over $\H_1$ and $\H_2$ ($I$ is the identity on $H\oplus H$ as vector spaces). Shale's criterion gives 
\[
 I i_1 - i_2  I \in\cL^2(\H_1 , \H_2)\, ,
\]
that entails the statement of the theorem by Prop. \ref{norm}. 
\endproof
We now consider the property
\ben\label{c5}
P_1 i_1|_H - P_2 i_2|_H \in \cL^2(H)\, ,
\een
that is
\ben\label{c1}
D^{-1}_1 -  D^{-1}_2 \in \cL^2(H)\, ,
\een
that is
\ben\label{c5b}
i_1\coth(L_1/2)|_H  - i_2\coth(L_2/2)|_H \in \cL^2(H)\, .
\een
We write $\a_1 \approx \a_2$ if Property \eqref{c5} holds. 
\bcor\label{cor0}
Assume $\a_1\approx \a_2$. 
The Gaussian states $\f_{\a_1}$ and $\f_{\a_2}$ are quasi-equivalent iff
\ben\label{c2}
D_2^{-1}\sqrt{1 + D_2^{2}}\,\Big(\sqrt{1 + D_1^{2}} - \sqrt{1 + D_2^{2}}\Big)\in \cL^2(H)
\een
and
\ben\label{c60}
\Big(\sqrt{1 + D_1^{2}} - \sqrt{1 + D_2^{2}}\Big)\in \cL^2(H)\, .
\een
\ecor
\proof
As $\a_1\approx \a_2$, i.e. $D^{-1}_1 -  D^{-1}_2 \in \cL^2(H)$, clearly \eqref{t1} is equivalent to
\ben\label{c4}
\sqrt{1 + D_2^{2}}\,\Big(D^{-1}_1\sqrt{1 + D_1^{2}} - D^{-1}_2\sqrt{1 + D_2^{2}}\Big)\in \cL^2(H)\, ,
\een
which is equivalent to \eqref{c2}. 

On the other hand, \eqref{t2} is equivalent to \eqref{c60}, again because $D^{-1}_1 -  D^{-1}_2 \in \cL^2(H)$. 
So the corollary follows by Thm. \ref{basic}. \endproof
\bcor\label{main}
Assume $\a_1\approx \a_2$.  
The Gaussian states $\f_{\a_1}$ and $\f_{\a_2}$ are quasi-equivalent iff
\ben\label{c16}
\Big(D_1^{-1}\sqrt{1 + D_1^{2}} - D_2^{-1}\sqrt{1 + D_2^{2}}\Big)\in \cL^2(H)
\een
and
\ben\label{c200D}
\Big(\sqrt{1 + D_1^{2}} - \sqrt{1 + D_2^{2}}\Big)\in \cL^2(H)\, .
\een
\ecor
\proof
Note first that, by \eqref{D12}, \eqref{c200D} is the same as 
\ben\label{c200}
\frac1{\cosh (L_1/2)}\Big|_H - \frac1{\cosh (L_2/2)}\Big|_H\in \cL^2(H)\, . 
\een
Let us now assume that $\a_1\approx \a_2$ and that \eqref{c200} holds. 
By Cor. \ref{cor0}, we have to prove that \eqref{c2} is equivalent to \eqref{c16}. 

By \eqref{P'i}, \eqref{c2} is equivalent to
\[
P'_2 i_2\Big(\frac1{\cosh (L_1/2)}\Big|_H - \frac1{\cosh (L_2/2)}\Big|_H\Big)\in \cL^2(H,\H_2)\, ,
\]
with  $P'_2$ the cutting projection $\H_2\to H$. 
As $P'_2 = 1 -P_2$, eq. \eqref{c2} is thus equivalent to
\ben\label{c12}
P_2 i_2\Big(\frac1{\cosh (L_1/2)}\Big|_H - \frac1{\cosh (L_2/2)}\Big|_H\Big)\in \cL^2(H)\,   ,
\een
namely
\ben\label{DDD}
\Big(D_2^{-1}\sqrt{1 + D_1^{2}} - D_2^{-1}\sqrt{1 + D_2^{2}}\Big)\in \cL^2(H)\, .
\een
Since $\sqrt{1 + D_1^{2}}$ is bounded, and $\a_1\approx \a_2$, the above equation is equivalent to \eqref{c16}.  
\endproof
\bcor\label{maincor}
The Gaussian states $\f_{\a_1}$ and $\f_{\a_2}$ are quasi-equivalent if
\ben\label{c161}
i_1\frac1{\sinh (L_1/2)}\Big|_H - i_2\frac1{\sinh (L_2/2)}\Big|_H\in \cL^2(H)\, .
\een
\ecor
\proof
Assume first that  $\a_1\approx \a_2$. Then \eqref{c161}, i.e. \eqref{c16}, is equivalent to \eqref{DDD}, and \eqref{DDD} implies \eqref{c200D} since $D_2$ is bounded. So Cor. \ref{main} applies and $\f_{\a_1}$ and $\f_{\a_2}$ are quasi-equivalent. 

To end our proof, we now show that  \eqref{c161} implies  $\a_1\approx \a_2$. 
Let $F$ be defined by $f(x) = F\big(g(x)\big)$, with $f(x) =  \coth(x)$, $g(x) = 1/\sinh(x)$. Then $f'(x) = F'(y)g'(x)$, with $y = g(x)$, so $F'(y) = f'(x)/g'(x) =  (1/\sinh^2(x))\big/(\cosh(x)/\sinh^2(x)) = 1/\cosh(x)$, therefore $F$ is uniformly Lipschitz. Since $0$ is not in the point spectrum of $L_k$, it follows by Cor. \ref{L2} that \eqref{c161} implies \eqref{c5b}, namely $\a_1\approx \a_2$. 
\endproof
Now, if $A_1, A_2$ are bounded, real linear operators on $H$ with trivial kernel, we have
\[
A_1  - A_2 = A_1(A^{-1}_2 - A^{-1}_1) A_2
\]
on the domain of the right hand side operator, thus
\ben\label{A-1}
A^{-1}_1  - A^{-1}_2 \in \cL^p(H) \Rightarrow A_1 - A_2 \in \cL^p(H)\, , \quad p \geq 1 \, .
\een
We then have:
\bcor
If 
\ben\label{c5b2}
i_1\coth(L_1/4)|_H  - i_2\coth(L_2/4)|_H \in \cL^2(H)\, ,
\een
then the Gaussian states $\f_{\a_1}$ and $\f_{\a_2}$ on $C^*(H)$ are quasi-equivalent. 
\ecor
\proof
By assumption \eqref{c5b2} holds, so also   
\ben
i_1\tanh(L_1/4) i_1 |_H - i_2\tanh(L_2/4) i_2 |_H \in \cL^2(H)\label{e1}\, ,
\een
holds by \eqref{A-1};
therefore
\[
i_1\big(\coth(L_1/4)|_H -  \tanh(L_1/4)|_H\big) 
- i_2\big(\coth(L_2/4)|_H -  \tanh(L_2/4)|_H\big) \in \cL^2(H)\, .
\]
Since $\coth(x/2) - \tanh(x/2) = 2/\sinh(x)$, we have
\ben\label{Ls}
i_1\frac1{\sinh (L_1/2)}\Big|_H - i_2\frac1{\sinh (L_2/2)}\Big|_H\in \cL^2(H)\, .
\een
So our corollary follows by Cor. \ref{maincor}. 
\endproof
The above corollary suggests that $\f_{\a_1}$ and $\f_{\a_2}$ are quasi-equivalent if $P_1 i_1|_H - P_2 i_2|_H$ is compact with proper values 
decaying sufficiently fast. 

\subsection{Weakly inner Bogoliubov automorphisms}
In this section, we study the condition for a real linear, symplectic bijection of a standard space to give rise to a weakly inner automorphism in the representation associated with a given Gaussian state. 

Let $H\subset\H$ be a factorial standard subspace of the complex Hilbert space $\H$, $T: H\to H$ a symplectic bijection and $\th_T$ the associated Bogoliubov automorphism of the Weyl algebra $A(H)$. 
Denote by $\A(H)$ the weak closure of $A(H)$ on $e^\H$ as in previous sections.  

We consider the real linear map on $\H$ given by 
\[
\hat T(h + h') = Th + h'\, , \quad h\in H,\, h'\in H'\, ,
\]
thus $D(\hat T) = {\rm ran}(\hat T) = H + H'$. One immediately sees that $\hat T$ is a symplectic map on $\H$. 

Note that $D([\hat T, i]) = D(\hat T)\cap iD(\hat T) = D(P_H)\cap D(P_{iH})$ is dense in $\H$, indeed a core for $P_H$, as in the proof of Lemma \ref{bdd}.  
\begin{lemma}\label{Ttilde0}
Let $T$ be a symplectic bijection on $H$.
The following are equivalent:
\begin{itemize}
\item[$(i)$] $\th_T$ extends to an inner automorphism of $\A(H)$;
\item[$(ii)$]  $\hat T^*\hat T - 1\in \cL^2(\H)$;
\item[$(iii)$]  $[\hat T, i]\in\cL^2(\H)$. 
\end{itemize}
\end{lemma}
\begin{proof}
Since $\A(H')$ is the commutant of $\A(H)$, $\th_T$ extends to an inner automorphism of $\A(H)$ if and only if the Bogoliubov automorphism associated with $\hat T$ is unitarily implementable on $e^\H$.
Therefore the equivalence $(i) \Leftrightarrow (ii)$ follows by Shale's criterion and Lemma \ref{ul}. 

$(ii) \Leftrightarrow (iii)$ follows again by Shale's criterion, Lemma \ref{ul} and the obvious adaptation of Lemma \ref{bdd}. 
\end{proof}
Set now $T = 1 + X$ and $\hat X = X \oplus 0$ on $H + H'$. In the symplectic matrix decomposition, we have
\begin{gather*}
\hat X i = \begin{bmatrix} X D^{-1}&   X D^{-1}\sqrt{1 + D^{2}} J\\
0 &   0 \end{bmatrix}\, ,
\\
i\hat X = \begin{bmatrix}  D^{-1}X&    0  \\
- J D^{-1}\sqrt{1 + D^{2}}\,  X&   0  \end{bmatrix}\, ,
\\
[\hat T,i]  = [\hat X,i]= \begin{bmatrix} [X,D^{-1}]&  X D^{-1}\sqrt{1 + D^{2}} J\\
J D^{-1}\sqrt{1 + D^{2}}X &  0 \end{bmatrix}\, ,
\end{gather*}
With $C = [\hat X,i]$, we apply Lemma \ref{En0}. 
Then
\begin{gather}\label{Inn1}
E_H C |_H  =  C_{11} + \sqrt{1 + D^{2}}\, J C_{21} =
[X,D^{-1}] + (D^{-1} + D)X\, ,
\\
\label{Inn3}
E_{H} C i|_{H'} =D C_{12} = D X D^{-1}\sqrt{1 + D^{2}}J\, ,
\\
\label{Inn2}
E_{H'} i C |_H =   J D J  C_{21} =  J   \sqrt{1 + D^{2}}X \, ,
\\
\label{Inn4}
E_{H'} C |_{H'}  =  J \sqrt{1 + D^{2}}\,  C_{12} +  C_{22} = 
J \sqrt{1 + D^{2}}\,   X D^{-1}\sqrt{1 + D^{2}} J \, .
\end{gather}
Note that
\[
D^{-1} + D = - i\big(\coth(L/2) - \tanh(L/2)\big) \big|_H = - i/\cosh(L/2)\sinh(L/2) \big|_H = - 2i/\sinh(L)\big|_H \, ,
\]
\[
D^{-1}\sqrt{1 + D^{2}} = - i \frac1{\sinh (L/2)}\Big|_H\, .
\]
\bprop\label{propinn}
$[\hat T,i] \in \cL^2(\H)$ iff all the operators 
\begin{gather*}
[X,D^{-1}] + (D^{-1} + D)X = XD^{-1} + DX
\, ,\\
D X D^{-1}\sqrt{1 + D^{2}}\, ,\\
\sqrt{1 + D^{2}}X\, ,\\
\sqrt{1 + D^{2}}\,   X D^{-1}\sqrt{1 + D^{2}} \, ,
\end{gather*}
are in $\cL^2(H)$. 

In particular, this is the case if $X D^{-1} \in \cL^2(H)$. 
\eprop
\proof
$[\hat T,i] \in \cL^2(\H)$ iff all the operators in \eqref{Inn1}, \eqref{Inn3}, \eqref{Inn2}, \eqref{Inn4} are Hilbert-Schmidt, so the first part of the statement holds. Now, $X D^{-1} \in \cL^2(H)$ implies that all the operators in the statement are Hilbert-Schmidt too as they are obtained by left/right multiplication of $X D^{-1}$ by bounded operators, $X D^{-1} \in \cL^2(H)$ is a sufficient condition for 
$[\hat T,i] \in \cL^2(\H)$. 
\endproof
\bthm\label{Boinn}
Let $(H,\a,\b)$ be an abstract factorial standard subspace and $T : H\to H$ a bijective symplectic map. Then $\th_T$ extends to an inner automorphism of the von Neumann algebra $\A(H)$, in the GNS representation of $\f_\a$ iff the conditions in Prop. \ref{propinn} hold. 
\ethm
\proof
The theorem follows now by Lemma \eqref{Ttilde0}. 
\endproof

\section{QFT and the modular Hamiltonian}
\label{QFT}
We now study the modular Hamiltonian in the free, massless, low-dimensional Quantum Field Theory. See \cite{BCM}, for numerical analysis in the massive case. 

\subsection{One-particle space of the free scalar QFT}
This section concerns the one-particle space of the free scalar QFT, especially in the low-dimensional case. Although we are here interested in the low-dimensional, massless case, we start by describing the general higher-dimensional case in order to clarify the general picture. 
In the following, $d$ is the space dimension, so $\RR^d$ is the time-zero space of the Minkowski spacetime $\RR^{d+1}$, cf. \cite{LM}. 

\subsubsection{Case $d\geq 2, m\geq 0$}
Let $\S$ denote the real linear space of smooth, compactly supported real functions on $\mathbb R^d$, $d\geq 2$. 

Let $H_m^{\pm1/2}$ be the  real Hilbert space of real tempered distributions $f \in S'(\bR^d)$ such that the Fourier transform $\hat f$ is a Borel function and
\ben\label{H12}
||f ||_{\pm1/2}^2  = \int_{\RR^d} {(|\bp|^2 + m^2)}^{\pm1/2} | \hat f(\bp)|^2 d\bp< +\infty\, .
\een
$\S$ is dense in $H_m^{\pm 1/2}$ and  $\mu_m : H_m^{1/2} \to H_m^{-1/2}$, with
\ben\label{mum}
\widehat{\mu_m f}({\bp}) = \sqrt{{|\bp|}^2 + m^2}\,\hat f({\bp})\, ,
\een
is a unitary operator. Then
\ben\label{imum}
\imath_m = \left[\begin{matrix}
0 & \mu_m^{-1} \\ -\mu_m  & 0  
\end{matrix}\right]
\een 
is a unitary operator $\imath_m$ on $H_m =  H_m^{1/2}\oplus H_m^{-1/2}$
with $\imath_m^2 = -1$,  namely a complex structure  on
$H_m$ that so becomes a complex Hilbert space $\H_m $ with the imaginary part of the scalar product given by
\ben\label{symp}
\Im( \langle f,g\rangle, \langle h,k \rangle )_m = \frac1 2\big( (h, g)-(f,k)\big) \, ,
\een
which is independent of $m\geq 0$ (where $(\cdot, \cdot)$ is the $L^2$ scalar product).

With $B$ the unit ball of $\mathbb R^{d}$, we shall denote by $H^{\pm1/2}_m(B)$ the  subspace of $H^{\pm1/2}_m$ associated with $B$ consisting of the distributions $f\in S'(\RR^d)$ as above that are supported in $B$. 
We have
\[
H^{\pm 1/2}_m(B) = \text{closure of $C_0^\infty(B)$ in $H^{\pm 1/2}_m$}\, ,
\]
and the standard subspace of $\H_m$ associated with $B$ is 
\[
H_m(B) \equiv H_m^{1/2}(B)\oplus H_m^{-1/2}(B)\, .
\]
Here $C^{\infty}_0(B)$ denotes the space of real $C^\infty$ function on $\mathbb R^d$ with compact support in $B$.

The $H_m(B)$'s, $m\geq 0$, are the same linear space with the same Hilbert space topologies (see e.g. \cite{LM}). We shall often identify these spaces as topological vector spaces. 

In the following, we consider the abstract standard spaces $(H,\a_m, \b)$ where $H= H_m(B)$, $\b$ is the symplectic form on $H$ given by \eqref{symp} and $\a_m$ is the real scalar product on $H$ as a real subspace of $\H_m$.

\subsubsection{Case $d=1$}
\label{d1}
$\bullet$ {\it Case $m >0$}. In this case the one-particle Hilbert space is defined exactly as in the higher dimensional case. In particular $
H_m^{\pm 1/2}$ is defined by \eqref{H12}
and  $\imath_m$ \eqref{imum}
is a complex structure on  $H_m= H^{1/2}_m\oplus H_m^{-1/2}$; so we have a complex Hilbert space $\H_m$, $m>0$. 
The subspace $H^{\pm 1/2}_m(B)$ of $H^{\pm 1/2}_m$ is again defined as in the higher dimensional case, with $B = (-1,1)$. 

We now set
\[
\dot H^{- 1/2}_m(B) = \text{closure of $\dot C_0^\infty(B)$ in $H^{-1/2}_m$},
\]
with 
\ben\label{dotS}
\dot\S = \Big\{f\in\S : \hat f(0) = \int_{\RR} f(x)dx =0\Big\} \, ,
\een 
$\dot C_0^\infty(B) =  C_0^\infty(B)\cap\dot\S$, and
\ben\label{H2}
\dot H_m(B) \equiv  H_m^{1/2}(B)\oplus \dot H_m^{-1/2}(B)  \, .
\een
\bprop\label{factdot}
$\dot H_m(B)$ is a standard subspace of
\ben\label{H1}
\dot\H_m \equiv \overline{\dot H_m(B) + \imath_m \dot H_m(B)}\,   .
\een
\eprop
\proof
As $\dot H_m(B)\subset H_m(B)$, clearly $\dot H_m(B)$ is separating, so
the statement is obvious. 
\endproof
\noindent
$\bullet$ {\it Case $m= 0$}. $H_0^{1/2}$ is defined as in the higher dimensional case \eqref{H12}:
\[
 H_0^{1/2} =\Big\{ f\in S'(\RR):\hat f \ \text{Borel function}\ \& \int_{\RR} |\bp| | \hat f(\bp)|^2 d\bp< +\infty\Big\}\,  .
\]
We now set
\[
\dot H_0^{-1/2} = \Big\{ f\in S'(\RR): \hat f \ \text{Borel function}\ \& \int_{\RR} |\bp^{-1}| |  \hat f(\bp)|^2 d\bp< +\infty\Big\}\,  .
\]
Note that
\[
\S\subset H^{\pm 1/2}_m\, , \ m > 0 \, ;\qquad \S\subset H^{1/2}_0\, ;\qquad \dot \S\subset \dot H^{-1/2}_0\, ,
\]
Then  $\imath_0$ (defined by \eqref{imum} with $m=0$)
is a complex structure on  $\dot H_0= H^{1/2}_0\oplus \dot H_0^{-1/2}$ and we get a complex Hilbert space $\dot \H_0$ with underlying real Hilbert space
$\dot H_0$. 

The subspace $H^{ 1/2}_0(B)$ of $H^{ 1/2}_0$ is defined as in the higher dimensional case. 
We also set
\[
\dot H^{- 1/2}_0(B) = \text{closure of $\dot C_0^\infty(B)$ in $\dot H^{-1/2}_0$},
\]
and
\ben\label{H20}
\dot H_0(B) \equiv  H_0^{1/2}(B)\oplus \dot H_0^{-1/2}(B)   \, .
\een
$\dot H_0(B)$ is a standard subspace of $\dot \H_0$. 
Note that, in the massless case, our {\it notation is unconventional:} $\dot\H_0$ is the usual one-particle space and $\H_0$ has not been defined yet. 
See also \cite{CM20,BFR21} for related structures. 

\subsection{The modular Hamiltonian, $d=1$}\label{MHd1}
We now describe the modular Hamiltonian associated with the unit double cone in the free, scalar QFT on the $1+1$ dimensional Minkowski spacetime. Recall that the modular Hamiltonian on the Fock space is the second quantisation of the modular Hamiltonian on the one-particle space, that will therefore be the subject of our analysis. 
In this subsection $B = (-1,1)$. 
\blem\label{equifact}
The $\dot H_m(B)$'s, $m\geq 0$, are the same linear space with the same Hilbert space topologies. 
Moreover, $\dot H_m(B)$ is a factorial standard subspace of $\dot\H_m$. 
\elem
\proof
The proof that the natural, real linear identifications of the $\dot H_m(B)$'s preserve the Hilbert space topology is a simple adaptation of the one given in the higher dimensional case, see \cite{LM}. 

We have seen in Prop. \ref{factdot} that $\dot H_m(B)$ is a standard subspace of $\dot\H_m$. 
The factoriality of $\dot H_0(B)$ follows, for example, by \cite{HL}. Now, the identification of $\dot H_m(B)$ with $\dot H_0(B)$ preserves the symplectic form. Since the factoriality is equivalent to the non-degeneracy of the symplectic form, also $\dot H_m(B)$ is factorial. 
\endproof
\blem\label{sympc}
$\dot H_m(B)'$, the symplectic complement of $\dot H_m(B)$ in $\dot\H_m$, is equal to $H_m(B)'\cap\dot \H_m$. 
\elem
\proof
The inclusion $H_m(B)'\cap\dot \H_m\subset \dot H_m(B)'$ is immediate. We prove the opposite inclusion. 
Let $f\oplus g\in \dot\H_m = H^{1/2}_m\oplus \dot H_m^{-1/2}$ belong to $\dot H_m(B)'$. By \eqref{symp},
\ben\label{fg0}
(h, g)-(f,k) = 0
\een
for all $h\oplus k \in\dot H_m(B)= H^{1/2}_m(B)\oplus \dot H_m^{-1/2}(B)$. 

Setting $k=0$, we see that $(h, g)=0$ for all $h\in C^\infty_0(B)$, so $g$ is supported in the complement $B^c$ of $B$, so $g\in H^{-1/2}_m(B^c)$ (for example by Haag duality). 

Set now $h=0$. Then $(f,k) = 0$ for all $k\in \dot H^{-1/2}_m(B)$. Let $F$ be the bounded linear functional on $H^{-1/2}_m(B)$
\[
F(k) \equiv (f,k) = \int f k\, ,\quad k\in H^{-1/2}_m(B)\, ;
\]
as $\dot H^{-1/2}_m(B)$ has codimension one in $H^{-1/2}_m(B)$, there exists $f_0\in H^{1/2}_m(B)$ such that, in particular,
\[
F(k) = \int f_0 k\, ,\quad k\in L^2(B)\, ,
\]
therefore $f_0 =0$. So $(f,k) = 0$ for all $k\in  C^\infty_0(B)$ and this implies $f \in H^{1/2}(B^c)$ by Haag duality. 
\endproof
Denote by $\dot P_m$  the cutting projection on  $\dot\H_m$ relative to $\dot H_m(B)$. 
\blem\label{cut2}
We have
\ben\label{PmM}
\dot P_m =
\left[\begin{matrix}  P_+ & 0 \\ 
0  &\dot P_-\end{matrix}\right]
\een 
with $P_+$ (resp. $\dot P_-$) the operator of multiplication by $\chi_B$ on $ H_m^{ 1/2}$ (resp. on  $\dot H_m^{- 1/2}$). 
\elem
\proof
Let $f\oplus g\in \dot\H_m = H^{1/2}_m\oplus \dot H_m^{-1/2}$ be in the domain of $\dot P_m$ and set $\dot P_m (f\oplus g) = f_0\oplus g_0 \in  \dot H_m(B)$. 
Thus $(f - f_0)\oplus (g - g_0)$ belongs to $\dot H_m(B)'$, the symplectic complement of $\dot H_m(B)$ in $\dot\H_m$; so, by Lemma \ref{sympc}, 
\[
(f - f_0)\oplus (g - g_0) \in H^{1/2}_m(B^c)\oplus \dot H^{-1/2}_m(B^c)
\]
and this shows that $\dot P_m$ is a diagonal matrix of the form \eqref{PmM}. 

We then have
\[
P_- g = g_0 = \chi_B g_0 = \chi_B\big( (g - g_0) + g_0\big) = \chi_B g \, .
\]
The equation $P_+ f = \chi_B f$, with $f$ in the domain of $P_+$, follows by similar arguments.   
\endproof
\subsubsection{$m = 0$}
In the massless case, the modular group associated with the unit, time-zero interval $B$ acts geometrically on the spacetime double cone spanned by $B$ \cite{HL}. We have:

\medskip\noindent
{\bf Theorem 5.9.}
{\it In the free scalar, massless, quantum field theory in $1+1$ spacetime dimension,
the modular Hamiltonian $\log \dot\Delta_{B,0}$  associated with the unit interval $B$, that is with the standard subspace $\dot H_0(B)\subset \dot\H_0$, is given by
\ben\label{Mmg}
\log \dot\Delta_{B,0}  = 2\pi \imath_0 \left[
 \begin{matrix}
0 & \frac12(1 - x^2) \\ \frac12(1 - x^2)\partial_x^2 - x\partial_x & 0
\end{matrix}\right] \,  .
\een
Setting $\log\dot \Delta_{B,0} = -2\pi \dot A_0$ and $\dot A_0 \equiv -\imath_0 \dot K_0$, we have that
$\dot K_0$ is essentially skew-selfadjoint on $\S\times\dot\S$. 
$\dot K_0^B = \dot K_0 |_{\dot H_0(B)}$ is skew-selfadjoint on $\dot H_0(B)$ and $C^\infty_0(B)\times \dot C^\infty_0(B)$ is a core for $\dot K_0^B$.}
\proof
The formula is obtained as in \cite{LM}, with obvious modifications. 
\endproof

\subsection{Local entropy of a wave packet, $d=1$} 
We shall be very short on the background for this section as this is explained in detail in \cite{CLR19,LM}. 

Let $\Phi$ be massless wave, $d=1$, with compactly supported, smooth Cauchy data $f,g$. 
Thus $\partial^2_t\Phi -\partial^2_x \Phi= 0$ and $f= \Phi|_{t=0}$, $g= \partial_t \Phi|_{t=0}$. 
The entropy $S_\Phi$ of $\Phi$ in $B = (-1,1)$ is given by
\[
S_\Phi = \Im(\Phi, P_H i \log\Delta_H\, \Phi) \ .
\]
Here, $H$ is the standard subspace $\dot H_0(B)\subset \dot\H_0$;
$\Delta_H$ is the modular operator and
$P_H$ is the cutting projection associated with $H$. $\Phi$ is the vector 
$f\oplus g\in \H_0 = H_0^{1/2}\oplus \dot H_0^{-1/2}$. The time-zero energy density of $\Phi$ is given by
$\langle T^{(0)}_{00}\rangle_{\Phi} = \frac12 \big( g^2  + (\partial_xf)^2   \big)$. 
\begin{theorem}\label{entropy}
The entropy $S_\Phi$ of the massless wave $\Phi$ in the unit interval $(-1,1)$ at time $t=0$ is given by
\ben\label{Sform}
S_{\Phi} = 2\pi\int^1_{-1} \frac{1 - x^2}{2} \langle T^{(0)}_{00}\rangle_{\Phi}\, dx \, .
\een
\end{theorem}
\proof
The proof follows the one in the higher dimensional case; this is possible as we now have the formula for the local modular Hamiltonian. 
\endproof
Note that the above results have a straightforward version with $B$ replaced by any other interval, same as \cite{LM}. 

\subsection{Further consequences in QFT}
In this section, we provide a direct consequence in second quantisation of our results. 
\subsubsection{Local entropy of coherent states}
By the analysis in \cite{L19,CLR19,LM}, we have an immediate corollary in Quantum Field Theory concerning the local vacuum relative entropy of a coherent state. 

Let $\A_0(B)$ be the von Neumann algebra associated with the unit space ball $B$ (thus to the causal envelope $O$ of $B$) by the free, neutral QFT on the Minkowski spacetime, $d\geq 1$, $m=0$. 
\begin{corollary}
Araki's relative entropy $S(\f_\Phi |\!| \f)$ on $\A_0(O)$ (see \cite{Ar76}) between the vacuum state $\f$ and the coherent state $\f_\Phi$ associated with the one-particle wave $\Phi\in\dot\H_0$ is given by \eqref{Sform}. 
\end{corollary}
\begin{proof}
The case $d\geq 2$ is proved in \cite{LM}.  By applying Theorem \ref{entropy}, the corollary follows now in the $d=1$ case too as in \cite{L19,CLR19}. 
\end{proof}

\section{Appendixes}
\subsection{Functional calculus for real linear operators}
The following proposition is part of Prop. 2.1 of \cite{LM}. Let $\B$ be the real algebra of complex, bounded Borel functions on $\RR$ such that  $f(-t) = \bar f(t)$
\bprop\label{inv}
Let $\H$ be a Hilbert space, $H\subset \H$ a closed, real linear subspace and $A :D(A)\subset \H\to \H$ a selfadjoint operator. With $K = iA$, the following are equivalent:
\begin{itemize}
\item[$(i)$]$e^{isA}H = H,   \ s\in \mathbb R$,
\item[$(ii)$] $f(A)H \subset H$,\ $f\in\B$,
\item[$(iii)$] 
 $D(K)\cap H$ is dense in $H$, $K(D(K)\cap H)\subset H$ and $K: (D(K)\cap H) \subset H\to H$ is skew-selfadjoint on $H$. 
\end{itemize}
\eprop
If $A$ and $H$ are as in Prop. \ref{inv}, we shall say that $H$ is $iA$-invariant. 

Let now $H$ be a real Hilbert space and $H_\bC$ the complexified Hilbert space, namely $H_\bC= H\oplus H$ with complex structure $\iota = \begin{bmatrix} 0 & -1\\ 1 & 0 \end{bmatrix}$. We write elements $x\in H_\bC$ as $x = \xi + \iota \eta$, $\xi, \eta\in H$. We have
\[
(\xi + \iota \eta , \xi' + \iota \eta') = (\xi, \xi') +  (\eta ,\eta') + i(\xi  , \eta')  - i (\eta , \xi' )\, ,
\]
\[
||\xi + \iota \eta||^2 = ||\xi||^2 + ||\eta||^2\, .
\]
 Let $T$ be a real linear, bounded operator on $H$. We denote by $\check T$ its promotion to $H_\bC$:
 \[
 \check T : \xi + \iota \eta \mapsto  T\xi + \iota T\eta\, ,
 \]
namely $\check T$ is the unique complex linear operator on $H_\bC$ that restricts to $T$ on $H$. Then $||\check T || = || T||$ because
\[
||\check T(\xi + \iota \eta)||^2 = ||T\xi||^2 + ||T\eta||^2 \leq ||T||(||\xi||^2 + ||\eta||^2) = ||T||\, ||\xi + \iota \eta||^2 \, .
\]
Note that \[
T\in \cL^2(H) \Leftrightarrow \check T\in \cL^2(H_\bC)\, ,
\]
indeed $||\check T||^2_2 = ||T||_2^2$ because a real orthonormal basis $\{e_k\}$ for $H$ is also a complex orthonormal basis for $H_\CC$ and
\[
||\check T||^2_2 = ||T||_2^2 = \sum_k ||Te_k||^2\, .
\]
Assume  that $T$ is skew-selfadjoint on $H$, namely $T^* = -T$. Then $\check T$ is skew-selfadjoint as complex linear operator on $H_\bC$, so $\iota \check T$  is a bounded selfadjoint operator on $H_\bC$. With $f$ a continuous complex function on $\bR$, we may define the complex linear operator $f(\iota \check T)$ on $H_\bC$ by the usual continuous functional calculus. 
Let then $f\in\B$; by Prop. \ref{inv} we have
\[
f(\iota \check T) H \subset H\, . 
\]
\bprop\label{PAX}
Let $H\subset\H$ be a standard subspace and $T$ a skew selfadjoint operator on $H$ as above. 
Suppose that
\ben\label{TX}
T = i X |_H
\een
with $X$ a selfadjoint operator on $\H$. With $A= -\iota \check T$ the selfadjoint operator on $H_\bC$ as above, we have 
\ben\label{AX}
 f(A)|_H =  f(X) |_H\, , 
\een
for every $f\in\B$. 
\eprop
\proof 
The statement holds if $f(x) = e^{ix}$ because $T$ is the infinitesimal skew-selfadjoint generator of $e^{isA}|_H = e^{isX}|_H$. So it holds if $f$ is the Fourier transform of a real $L^1$-function $g$ as
\[
 f(A)|_H = \int g(s) e^{-isA}|_H ds = \int g(s) e^{-isX}|_H ds = f(X)|_H 
\]
Then \eqref{AX} holds for every continuous function with compact support $f\in\B$, as it can be uniformly approximated by functions as above by the Stone-Weierstrass theorem. 

Let now $f$ be any function in $\B$ and fix two vectors $\xi,\eta\in H$. There exists a uniformly bounded sequence of continuous functions $f_n\in\B$ with compact support such that $f_n\to f$ almost everywhere with respect to the spectral measures of $A$ and $X$ associated with $\xi,\eta$. Then
\[
(\xi,  f(A)\eta) = \lim_n (\xi,  f_n(A)\eta) = \lim_n (\xi, f_n(X)\eta) = (\xi, f(X)\eta)
\]
by the Lebesgue dominated convergence theorem, that concludes our proof because $\xi,\eta$ are arbitrary.

\subsection{Operator Lipschitz perturbations} 
The next theorem is due to Potatov and Sukochev \cite{PS11}.  
\bthm\label{PS}
Let $A_1, A_2$ be selfadjoint operators on a Hilbert space $\H$ and $f$ a uniformly Lipschitz function on $\RR$.   
If $A_1 - A_2\in\cL^p(\H)$, with $p >1$,
then also $f(A_1) - f(A_2)\in\cL^p(\H)$. 
\ethm
\noindent
Note that, in Thm. \ref{PS}, it suffices to assume that $(A_1 - A_2)|_\D\in\cL^p(\H)$ with $\D$ a core for $A_1$ or $A_2$, since then $\D$ is a core for both $A_1$ or $A_2$ and $D(A_1) = D(A_2)$ because $A_1 - A_2$  is bounded. 

The following corollary was communicated to us by F. Sukochev. 
\bcor\label{PScor}
Let $A_k$ be a selfadjoint operator on the Hilbert space $\H_k$, $k=1,2$, and suppose that $\H_1$ and $\H_2$ are the same topological vector space, that we call $\H$. Then
\[
A_1 - A_2\in\cL^p(\H)\implies f(A_1) - f(A_2)\in\cL^p(\H)\, ,
\]
$p>1$, for every uniformly Lipschitz function $f$ on $\RR$.   
\ecor
\proof
Let $C: \H_1 \to \H_2$ be the complex linear identification of $\H_1$ and $\H_2$ as topological vector spaces. So $C$ is a bounded operator with bounded inverse $C^{-1}$. Then we have to show that 
\[
A_1 - C^{-1} A_2 C\in\cL^p(\H_1) \implies f(A_1) - C^{-1}  f(A_2) C\in\cL^p(\H_1)\, ,
\]
or, equivalently, that
\[
CA_1 -  A_2 C\in\cL^p(\H_1,\H_2) \implies C f(A_1) -  f(A_2) C\in\cL^p(\H_1,\H_2)\, .
\]
With $\K = \H_1\oplus \H_2$, the operator $A = A_1\oplus A_2$ is selfadjoint on $\K$. Set $V = \begin{bmatrix} 0 & 0\\C & 0\end{bmatrix}$; then
\[
V A  - A V= \begin{bmatrix} 0 & 0\\CA_1 -  A_2 C & 0\end{bmatrix}
\]
and
\[
V f(A) -  f(A) V = \begin{bmatrix} 0 & 0\\Cf(A_1) -  f(A_2) C & 0\end{bmatrix}\, ,
\]
so we have to show that 
\[
V A - A V\in\cL^p(\K)\implies V f(A)  -  f(A) V\in\cL^p(\K)\, ,
\]
that follows by \cite[Eq. (14)]{PS11}. 
\endproof
\bcor\label{L2}
Let $H_k\subset \H_k$ be a standard subspace and $X_k$ a selfadjoint operator on $\H_k$ such that $H_k$ is $i_k X_k$-invariant, $k =1,2$. 
Suppose that $H_1$ and $H_2$ are the same real linear space $H$ with equivalent scalar products. 
Then
\[
i_1 X_1|_H - i_2 X_2|_H \in \cL^p(H) \implies
i_1 f(X_1)|_H - i_2 f(X_2)|_H \in \cL^p(H)  \, ,
\]
$p>1$, for every uniformly Lipschitz function $f$ on $\RR$ such that $f(-x) = -\overline{f(x)}$. 
\ecor
\proof
Let ${H_k}_\CC$ be the usual complexification of the real Hilbert space $H_k$. Then ${H_1}_\CC$ and ${H_2}_\CC$ are equivalent complex Hilbert spaces. 

Let $A_k$ be the selfadjoint extension of $X_k$ to ${H_k}_\CC$ as above; by Prop. \ref{PAX},  we have
\begin{multline*}
i_1 X_1|_H - i_2 X_2|_H\in \cL^p(H)\implies
A_1 - A_2  \in \cL^p(H_\CC)\implies
\iota f(A_1) - \iota f(A_2)  \in \cL^p(H_\CC)\\ \implies
\iota f(A_1)|_H - \iota f(A_2)|_H \in \cL^p(H)\implies
i_1 f(X_1)|_H - i_2 f(X_2)|_H  \in \cL^p(H) \, .
\end{multline*}
\endproof

\subsection{Extensions of the Laplacian via Helmholtz operator}
\label{Laplacian}
Let $\H$ be a Hilbert space, $\K$ a closed subspace and $A: D(A)\subset \H\to \H$ a positive selfadjoint linear operator. Assume that
\[
D_0 = \big\{\x\in D(A)\cap \K : A\x\in\K\big\}
\]
 is dense in $\K$ and denote by $A_0$ the restriction of $A$ to $D_0$, as operator $\K \to\K$. Clearly $A_0$ is a positive Hermitian operator on $\K$. We want to study the selfadjoint extensions of $A_0$. 

Choose $m>0$, then $(A + m^2)^{-1}$ is a bounded selfadjoint operator on $\H$ whose norm is $||(A + m^2)^{-1}||\leq 1/m^2$.  
With $E$ the orthogonal projection of $\H$ onto $\K$, set
\ben\label{TA}
T = E(A +m^2)^{-1}|_\K\, .
\een
Then $T$ is a bounded, selfadjoint operator on $\K$ and $||T||\leq 1/m^2$. We have
\ben\label{BA}
T(A_0 + m^2) \xi = \xi\, ,\quad \xi\in D_0\, .
\een
We note the following.
\smallskip

\noindent
$\bullet$ $\ker(T) = \{0\}$. Let $\xi\in \K$; since $T\xi = 0$ implies
\[
(\xi, T\xi) =  (\xi, E(A + m^2)^{-1}\xi) = (\xi, (A +m^2)^{-1}\xi) = ((A +m^2)^{-1/2}\xi, (A +m^2)^{-1/2}\xi)=0\, ,
\]
we have
\[
T\xi = 0 \implies  (A +m^2)^{-1/2}\xi = 0 \implies \xi = 0\, .
\]
$\bullet$ Let $A_m$ be defined by $ (A_m + m^2) \equiv T^{-1}$. Then $A_m$ is a positive, selfadjoint extension of $A_0$ on $\K$ and $A_m \geq m^2$.  Indeed, eq. \eqref{BA} implies
\[
T^{-1}\xi = (A_0 + m^2)\xi \, , \quad \xi\in D_0\, .
\]
\smallskip

\noindent
$\bullet$ By theorems of von Neumann, Krein, Friedrichs et al. (see \cite{AS,Sch}), every positive selfadjoint extension of $A_0$ lies between $A_{\min}$ and $A_{\max}$, where where $A_{\min}$ and $A_{\max}$ are respectively the Krein and the Friedrichs extension of $A_0$ on $\K$. 
In particular,
\ben\label{minmax}
A_{\min} \leq A_m \leq A_{\max}\, ,
\een
in the quadratic form sense.   

Consider now the case of $\K = L^2(B)\subset \H = L^2(\RR^d)$. If $f\in \C^\infty(\partial B)$, consider the exterior Dirichlet problem for the Helmholtz operator: find a smooth function $f^c$ on the complement $B^c$ of $B$ such that:
\[
f^c |_{\partial B} = f\, , \quad  (\nabla^2 - m^2)f^c = 0 \ \text{on the complement of}\  \bar B\, ;
\]
this problem is studied e.g. \cite{MS}. 

Denote by $C_m$ the space of all $f\in C^\infty(\partial B)$ such that $f^c$ exists with $f^c$ and the partial derivatives of all order tending to zero  as $r = |x| \to +\infty$ faster than any inverse power of $r$.  In this case the solution $f^c$ is unique by the maximum principle. 

We sketch the following proposition. 
\bprop\label{Aext}
Let $\H = L^2(\RR^d)$, $\K = L^2(B)$, and $A = -\nabla^2$  be the Laplacian on $L^2(\RR^d)$; then 
\[
A_m =  -\nabla^2_m\, ,
\]
where $\nabla^2_m$ is the Laplacian on $L^2(B)$ with boundary condition
\[
\partial^-_r f  = - \partial^+_r f^c\ {\rm on}\ \partial B\, ,  
\]
more precisely, $D_m \equiv \big\{f\in C^\infty(\bar B) : f|_{\partial B}\in C_m,\ \partial^-_r f = -\partial^+_r f^c\ {\rm  on}\ \partial B\big\}$ is a core for $A_m$, with $\partial_r^\pm$ denoting the outer/inner normal derivative. 
\eprop
\proof
Let $g\in C_0^\infty(B)$ and $f = (A +m^2)^{-1}g$. Then $f\in D(\nabla^2)$ and $f$ is a solution of the equation 
$
(-\nabla^2 + m^2)f = g
$
on $\RR^d$. In particular $(-\nabla^2 + m^2)f = 0$ on $B^c$, namely 
$f|_{B^c} = {(f |_{\partial B}})^c$. As $g\in C_0^\infty(B)$, $f$ belongs to the Schwarz space $S(\RR^d)$, thus $f|_{B^c}\in C_m$. 

With $T$ given by \eqref{TA}, we have $Tg = f|_{\bar B}$; as $T$ is a bounded operator on $L^2(B)$ and $C_0^\infty(B)$ is dense in $L^2(B)$,
the domain $T C_0^\infty(B)$ is a core for $A_m = T^{-1}$. Since $T C_0^\infty(B)\subset D_m$, we have that $A_m$ is essentially selfadjoint on $D_m$. Clearly,
$A_m = -\nabla_m^2$ on $D_m$. 

Now $-\nabla_m^2$ is Hermitian on $D_m$ by the Green identity (consider the integration on the boundary of a corona $1\leq r\leq R$ and then let $R\to
\infty$), so we conclude that $A_m = -\nabla_m^2$ because selfdajoint operators are maximal Hermitian.  
\endproof
\noindent
{\bf Acknowledgements.} 
We thank R. Conti and G. Morsella for several valuable comments, also in relation to their work in progress, F. Sukochev for providing Corollary \ref{PScor}, D. Bahns, K.H. Rehren and the Alexander von Humboldt Foundation for the invitation at the University of G\"ottingen during November 2021, where the final part of this paper has been written, and D. Buchholz for stimulating discussions. 
\smallskip

\noindent
We acknowledge the MIUR Excellence Department Project awarded to the Department of Mathematics, University of Rome Tor Vergata, CUP E83C18000100006.
\bigskip

\noindent
Data sharing not applicable to this article as no datasets were generated or analysed during the current study.

\end{document}